\documentclass[runningheads]{llncs}

\usepackage{graphicx}
\usepackage[ruled,vlined]{algorithm2e}
\usepackage[title]{appendix}
\usepackage{amsmath}
\usepackage{array}
\usepackage{multirow}
\SetKwComment{Comment}{$\triangleright$\ }{}
\usepackage{nicefrac}

\begin{document}
\title{An Algorithm for Reordering Buffer Management Problem and Experimental Evaluations on Discrete Distributions}

\titlerunning{An Algorithm for RBM Problem and Experimental Evaluations}

\author{Gözde Filiz\inst{1}\orcidID{0000-0003-4891-5122} \and
M. Oğuzhan Külekci\inst{1}\orcidID{0000-0002-4583-6261}}

\authorrunning{G. Filiz \and M. O. Külekci}

\institute{Informatics Institute, Istanbul Technical University, Istanbul, Turkey \email{filiz18@itu.edu.tr} \\ \email{kulekci@itu.edu.tr}}

\maketitle             

\begin{abstract}
In the reordering buffer management problem, a sequence of requests must be executed by a service station, where a cost occurs for each pair of consecutive requests with different attributes. A reordering buffer management algorithm aims to permute the input sequence using the buffer to minimize the total cost. Reordering buffers has many potential applications in computer sciences and economics. In this article, we proved the minimum buffer length for the optimal solution to the reordering buffer management problem in the offline setting.  With the assumption that color selection is always made when the buffer is full, selecting the most frequent color from the buffer given the smallest buffer size $k$ that satisfies either $o_1 < 2 \cdot \lceil \nicefrac{k}{\sigma} \rceil$ \textsc{OR} $o_2 < \lceil \nicefrac{k}{\sigma} \rceil$ guarantees the optimal solution, where $o_1$ and $o_2$ represent respectively the frequency of the most and the second most frequent colors in the input sequence $\mathcal{X}$, and $\sigma$ is the number of distinct colors appearing in $\mathcal{X}$. We proposed a new algorithm for the online setting of the problem that uses the results of the proof made on the minimum buffer length required for the optimal solution. Moreover, we presented the results of the first experimental setup that uses input sequences following discrete distributions to evaluate the performance of algorithms. Out of 432 cases, the new algorithm showed the best performance in 409 cases that is approximately $95\%$ of all cases.

\keywords{Scheduling algorithms \and Online algorithms \and Reordering\\buffer management problem.}

\end{abstract}

\section{Introduction}\label{sec:introduction}

Scheduling problems play an important role in a broad range of application areas such as economics, computer systems, and everyday life. In the scheduling problem, tasks or requests have to be assigned to a given resource while satisfying some precedence constraints. For some applications, in which tasks can be delayed for a certain amount of time, buffers can be used to store the delayed items temporarily to permute the tasks in favor of precedence constraints. Each arriving task is stored in the buffer until the buffer is full. When the buffer is filled with tasks, an effective buffer management strategy must answer important questions as which tasks to store and which tasks to execute at any time according to the problem’s specific constraints. 

In hard disks, a reordering buffer can be used to categorize accesses according to their destination cylinder to minimize the movement of the head. This problem is known as disk scheduling. A comparative analysis of disk scheduling policies is studied by Teorey and Pinkerton \cite{teorey1972comparative} in 1972.
In computer graphics, a reordering buffer can be used to rearrange the incoming sequence of primitives to reduce the cost of state changes. Krokowski, Räcke, Sohler, and Westermann \cite{DBLP:conf/vmv/KrokowskiRSW04} presented experiments revealing that a pipeline buffer reduces the number of state changes by an order of magnitude and reduces the rendering time by roughly $30\%$.
A network communication system can benefit from a reordering buffer management strategy by permuting the packets according to their destination nodes. 

\subsection{Problem Definition}
In the reordering buffer management problem, we have an input sequence of $n$ items that are characterized by a particular attribute, for simplicity we will call this attribute the color of the item. All items are waiting in queue to be processed by a service station. Different processes will be applied to the items according to their color. Changing the process in the service station creates a cost. Therefore, a cost occurs for each pair of consecutive items with different colors. To minimize the cost, we use a buffer of size $k$. Items waiting in the queue are first inserted into the buffer. When the buffer is full, the reordering buffer management strategy selects an output color and all items of that color in the buffer are forwarded to the service station. Items are added to the output sequence after they are processed. The vacancies in the buffer are filled with the items waiting in the queue. This mechanism continues until all items of the input sequence are processed. All algorithms start and end with an empty buffer. 
As stated by Räcke, Sohler, and Westermann \cite{DBLP:conf/esa/RackeSW02}, reordering buffer management strategies are evaluated using lazy strategies which fulfill the following properties.
\begin{itemize}
  \item If there exists an item in the buffer with the current output color, then the algorithm does not perform a switch in the decision.
  \item If there is a space in the buffer, the algorithm does not evict any item from the buffer.
\end{itemize}
Thus, a reordering buffer management strategy has only to specify the selection of the output color. Moreover, every strategy can be transformed into a lazy strategy without increasing the cost.

Let $\mathcal{X} = \langle 1, 2, 2, 1, 3, 3, 3, 2, 3, 2, 2  \rangle$ represent a sequence of 11 items to be processed by service station with buffer size $k=5$ where color of items are represented as numbers. Assume the reordering buffer management strategy simply chooses the most frequent color in the buffer. Figure \ref{a_simple_process} shows the process of this input sequence. The process starts with an empty buffer. In stage 2, items that are waiting in the queue are added to the buffer until the buffer is full. Reordering buffer management strategy selects color 1 as output color in stage 3 and the items of output color in the buffer are forwarded to the service station. In stage 5, vacancies in the buffer are filled with items waiting in the queue. Since there is no item of color 1 in buffer, reordering buffer management strategy selects a new color, 3. After items of color 3 are forwarded to the service station, the buffer is filled again in stage 6. Since there is an item of the last output color within the newcomers, the strategy does not perform a switch in the decision and the item is forwarded to the service station. In stage 7, the strategy selects color 2 as output color, and all items of color blue are forwarded to the service station. In stage 8, the process ends since there is no item left in the queue or buffer. The output sequence is $\mathcal{Y} = \langle 1, 1, 3, 3, 3, 3, 2, 2, 2, 2, 2 \rangle$. The number of color switches in the input sequence was 6 and is reduced to 2 in the output sequence. Therefore, the process eliminated the cost that will occur for each pair of consecutive items with different colors by 4 times. 
\begin{figure}
\scriptsize
\centering
\begin{tabular}{ |r c c c l c l| }
\hline
\multicolumn{7}{|c|}{\bfseries Stage 1} \\

Input sequence &  & Buffer &  & Service Station &  & Output Sequence\\
\hline
2, 2, 3, 2, 3, 3, 3, 1, 2, 2, 1 & {\bfseries $\rightarrow$} & & {\bfseries $\rightarrow$} & & {\bfseries $\rightarrow$} & \\
\hline
\hline
\multicolumn{7}{|c|}{\bfseries Stage 2} \\

Input sequence &  & Buffer &  & Service Station &  & Output Sequence\\
\hline
2, 2, 3, 2, 3, 3 & {\bfseries $\rightarrow$} & 3, 1, 2, 2, 1 & {\bfseries $\rightarrow$} & & {\bfseries $\rightarrow$} & \\
\hline
\hline
\multicolumn{7}{|c|}{\bfseries Stage 3} \\

Input sequence &  & Buffer &  & Service Station &  & Output Sequence\\
\hline
2, 2, 3, 2, 3, 3 & {\bfseries $\rightarrow$} & 3, 1, 2, 2 & {\bfseries $\rightarrow$} & 1 & {\bfseries $\rightarrow$} & \\
\hline
\hline
\multicolumn{7}{|c|}{\bfseries Stage 4} \\

Input sequence &  & Buffer &  & Service Station &  & Output Sequence\\
\hline
2, 2, 3, 2, 3, 3 & {\bfseries $\rightarrow$} & 3, 2, 2 & {\bfseries $\rightarrow$} & 1 & {\bfseries $\rightarrow$} & 1 \\
\hline
\hline
\multicolumn{7}{|c|}{\bfseries Stage 5} \\

Input sequence &  & Buffer &  & Service Station &  & Output Sequence\\
\hline
2, 2, 3, 2, & {\bfseries $\rightarrow$} & 3, 3, 3, 2, 2 & {\bfseries $\rightarrow$} &  & {\bfseries $\rightarrow$} & 1, 1 \\
\hline
\hline
\multicolumn{7}{|c|}{\bfseries Stage 6} \\

Input sequence &  & Buffer &  & Service Station &  & Output Sequence\\
\hline
2 & {\bfseries $\rightarrow$} & 2, 3, 2, 2, 2 & {\bfseries $\rightarrow$} &  & {\bfseries $\rightarrow$} & 3, 3, 3, 1, 1 \\
\hline
\hline
\multicolumn{7}{|c|}{\bfseries Stage 7} \\

Input sequence &  & Buffer &  & Service Station &  & Output Sequence\\
\hline
 & {\bfseries $\rightarrow$} & 2, 2, 2, 2, 2 & {\bfseries $\rightarrow$} &  & {\bfseries $\rightarrow$} & 3, 3, 3, 3, 1, 1 \\
\hline
\hline
\multicolumn{7}{|c|}{\bfseries Stage 8} \\

Input sequence &  & Buffer &  & Service Station &  & Output Sequence\\
\hline
 & {\bfseries $\rightarrow$} &  & {\bfseries $\rightarrow$} &  & {\bfseries $\rightarrow$} & 2, 2, 2, 2, 2, 3, 3, 3, 3, 1, 1 \\
\hline
\end{tabular}
\caption{A simple process with buffer management strategy.}\label{a_simple_process}
\end{figure}

\noindent To formally state the reordering buffer management problem various approaches can be used. From the perspective of metric spaces, we are given an input sequence of requests for service each of which corresponds to a point in a metric space $(V,d)$ that is, a request $r$ corresponds to a point $p(r) \in V$, where $V$ is a set of points and $d$ is a distance function. Serving a request $r$ following the service to a request $q$ induces cost $d(p(r),p(q))$, that is, the distance between the points corresponding to the two requests.  At each point in time, the reordering buffer contains the first $k$ requests of the input sequence that have not been served so far. A scheduling strategy has to decide which request to serve next. Upon its decision, the corresponding request is removed from the buffer and appended to the output sequence, and thereafter the next request in the input sequence takes its place. \cite{englert2008online}

A more abstract formulation can be made as; an input sequence $\sigma = \sigma_1\sigma_2\ldots\sigma_l$ of requests that correspond to points in a metric space $(V,d)$ is given. This input sequence can be reordered such that an output sequence $\sigma_\pi= \sigma_{\pi(1)}  \sigma_{\pi(2)}\ldots\sigma_{\pi(l)}$ that is a permutation $\pi$ of the input sequence satisfying $\pi(i) < i+k$ is produced. The objective is to minimize the cost of the output sequence
$\Sigma_{i=1}^{l-1} d(p(\sigma_{\pi(i)}),\\p(\sigma_{\pi(i+1)}))$. This problem formulation is proved to be indeed equivalent by Englert and Westermann \cite{DBLP:conf/icalp/EnglertW05}.

From the perspective of graph theory, an undirected weighted graph $G=(V,E)$ and an input sequence of requests which correspond to vertices in $G$ is given. At each point in time, the first $k$ unprocessed requests from the input sequence are located at their respective vertices in the graph. The scheduling algorithm controls one server that moves through the graph. To serve a request, the server has to move to the vertex containing the request. The cost of this movement is given by the length of the chosen path. When a request is served it is removed from the graph and the next request from the input sequence is placed at its corresponding vertex. The objective is to minimize the total distance the server moves. \cite{englert2008online}

The reordering buffer management problem can be classified according to its components. One of the defining characteristics of the problem is the cost function. The cost can be uniform or non-uniform. In the uniform case, the same cost occurs when switching from one color to another. Therefore, the cost can be measured as the number of color changes in the output sequence. In the non-uniform case, different costs can occur when switching from one color to another. In the maximization variant of the problem, instead of minimizing the number of color changes in the output sequence, the goal is stated as maximizing the number of color changes that are eliminated from the input sequence. In the online variant of the problem, the algorithm has to decide without the knowledge of the items that will arrive in the future. In other words, the input sequence is unknown in advance and the algorithm has to work with the partial knowledge of the input sequence. On the other hand, an offline algorithm can make its decisions based on looking at all items in the input sequence. The performance of online algorithms is evaluated using competitive analysis proposed by Sleator and Tarjan \cite{sleator1985amortized}. In the competitive analysis, an online algorithm is compared against an optimal offline algorithm. Let $C_A (\sigma)$ denote the cost of the solution given by algorithm A on the input sequence $\sigma$ and let OPT denote an optimal offline algorithm. For minimization problems, an online algorithm A is said to be c-competitive if and only if $C_A (\sigma)  \leq c \cdot C_{OPT} (\sigma)  + x$  for every input sequence $\sigma$, where $x$ is a term that does not depend on $\sigma$. Similarly, an online algorithm A for a maximization problem is c-competitive if and only if  $C_A (\sigma)  \geq c \cdot C_{OPT} (\sigma)  - x$ for every input sequence $\sigma$. We call the smallest c such that A is c-competitive the competitive ratio of A and we say that A achieves a competitive ratio of c if A is c-competitive. If $x$ can be chosen to be 0, A is said to be strictly c-competitive. 

\subsection{Related Work}
The reordering buffer management problem was introduced by Räcke, Sohler, and Westermann \cite{DBLP:conf/esa/RackeSW02} in 2002. They proved a lower bound of $\Omega(\sqrt{k})$ on the competitive ratio of the First in First Out (FIFO) and Last Recently Used (LRU) strategy and of $\Omega(k)$ on the competitive ratio of the Largest Color First (LCF) strategy. Their main result is a deterministic strategy called Bounded Waste (BW) which achieves a competitive ratio of $O(\log^2{k})$.

Englert and Westermann \cite{DBLP:conf/icalp/EnglertW05} proposed a new strategy called Maximum Adjusted Penalty (MAP) for the non-uniform cost which has a competitive ratio of $O(\log(k))$. In their work, they also prove a basic upper bound for both online and offline strategies in arbitrary metric spaces as the reordering buffer of size $k$ reduces the cost of each sequence by a factor of at most $2k-1$. Rabani and Elgrabli \cite{DBLP:conf/soda/Avigdor-ElgrabliR10} proposed Threshold or Lowest Cost (TLC) algorithm which has a competitive ratio of $O(\nicefrac{\log{k}}{\log\log{k}})$ with non-uniform cost. The algorithm maintains and updates a counter for every item. New color selection is made by a comparison between the counters and the cost of colors.

Adamaszek, Czumaj, Englert, and Räcke \cite{DBLP:conf/stoc/AdamaszekCER11} proved that deterministic online algorithms for the problem have a competitive ratio of at least $\Omega(\sqrt{\nicefrac{\log{k}}{\log\log{k}}})$ and randomized online algorithms have a competitive ratio of at least $\Omega(\log\log{k})$. They also proposed a deterministic online algorithm, called Largest Color Class (LCC), for the problem with a competitive ratio of $\Omega(\sqrt{\log{k}})$. 

Rabani and Elgrabli \cite{DBLP:conf/focs/Avigdor-ElgrabliR13} proposed an $O(\log\log{k})-$competitive randomized algorithm for the problem with uniform cost, matching the lower bound proved by Adamaszek, Czumaj, Englert, and Räcke \cite{DBLP:conf/stoc/AdamaszekCER11}. The algorithm consists of linear and dual programs that work in two phases running in parallel. Extending the randomized algorithm to cover the non-uniform cost, Im and Moseley \cite{DBLP:conf/soda/ImM14} proposed a new algorithm that uses linear programming relaxation techniques. Their algorithm has a competitive ratio of $O(\log\log{k\gamma})$, where $\gamma$ is the ratio of the maximum to minimum cost in the input sequence. 

Adamaszek, Renault, Rosén, and Stee \cite{DBLP:journals/scheduling/AdamaszekRRS17} proposed an intriguing algorithm for the online version of the problem that uses advice bits to choose the new color to be evicted. They proved that for any $\epsilon > 0$ there is a $(1 + \epsilon)-$competitive algorithm for the problem that uses only a constant (depending on $\epsilon$) number of advice bits per input item. Moreover, they presented lower bound $\Omega(\log{k})$ bits of advice per request for an optimal deterministic algorithm. 

A variant of the problem is proposed by Azar, Englert, Gamzu, and Kidron \cite{DBLP:conf/stacs/AzarEGK14} that consists of a service station that has $k$ servers and a buffer of size $b$. An item can be served by any server that has the configuration of the color. They proposed a randomized online algorithm that achieves a competitive ratio of $O(\sqrt{b}\ln{k})$.

Khandekar and Pandit \cite{DBLP:conf/stacs/KhandekarP06} studied the problem on line metric, which corresponds to the application on disc scheduling. They proposed a randomized online algorithm with a competitive ratio of $O(\log^3{n})$ for a disc with $n$ uniformly spaced tracks. Moreover, they showed that natural strategies such as First in First Out, Nearest First, and Scan are not appropriate for this problem and they have a competitive ratio of $\Omega(\log{k})$. Gamzu and Segev \cite{DBLP:conf/stacs/GamzuS07} proposed a deterministic $O(\log{n})$-competitive algorithm called Moving Partition for n-point evenly spaced line metrics.

Chan, Megow, Stee, and Sitters \cite{DBLP:journals/corr/abs-1009-4355} proved the offline version of the problem with uniform cost is NP-hard by a reduction from the 3-Partition problem. Concurrently, Asahiro, Kawahara, and Miyano \cite{DBLP:journals/dam/AsahiroKM12} gave proof for NP-hardness of the problem on the uniform metric by a reduction from the Vertex Cover problem. Khandekar and Pandit \cite{DBLP:conf/isaac/KhandekarP06} considered the offline version of the problem for the line metrics. They proposed an algorithm with a constant factor approximation that runs in quasi-polynomial time which is based on a dynamic program.

Rabani and Elgrabli \cite{DBLP:conf/soda/Avigdor-ElgrabliR13} proposed the first constant-factor approximation algorithm for the offline version of the problem on general metric spaces. Their algorithm is based on an intricate rounding of the solution to a linear programming relaxation of the problem.

For the maximization variant of the offline version of the problem, Kohrt and Pruhs \cite{DBLP:conf/latin/KohrtP04} proposed a polynomial-time $\nicefrac{1}{20}-$approximation algorithm. Improving their results, Bar-Yehuda and Laserson \cite{DBLP:conf/waoa/Bar-YehudaL05} proposed a polynomial-time $9-$approximation algorithm. In their approach, a profit is gained and assigned to an item if it is followed by another item of the same color at the service station.

Several strategies for uniform cost are evaluated by Englert, Röglin, and Westermann \cite{DBLP:conf/wea/EnglertRW06}. The strategies MAP, FIFO, LRU, and Most Common First (MCF) used on random input sequences together with two variants of MAP, namely Round Choice (RC) and Round Robin (RR). Since the MAP strategy is computationally very intensive, they presented an adapted version of the strategy which is similar to the BW strategy. For the evaluation of the strategies, they introduced a method to generate input sequences with known optimal solutions. For random sequences with constant block lengths, the experiments revealed that MCF and MAP achieve the best competitive ratios, MCF is optimal for buffer size greater than 49, and MAP is also optimal for buffer sizes greater than 317. For random sequences with uniformly chosen block lengths, the experiments indicate that the competitive ratio of FIFO, LRU, and MCF increases with the buffer size where MAP, RC, and RR achieve a small constant of competitive ratios.

A different experimental evaluation is presented by Krokowski, Räcke, Sohler, and Westermann \cite{DBLP:conf/vmv/KrokowskiRSW04} on the pipeline buffer scenario which is an application of the problem on line metric. An input sequence of primitives with a specific attribute has to be arranged by a limited-sized buffer, to minimize the number of state changes. Their experiments revealed that a pipeline buffer reduces the number of state changes by an order of magnitude, reduces the rendering time by roughly $30\%$, and achieves almost the same rendering time as optimal. The experiments indicated that RR achieves the best performance among other strategies improving between $28.8\%$ and $35\%$ on the rendering time without sorting buffer for the textured scenes.

Vasanth \cite{Vasanth2015} advanced the experiments of Englert, Röglin, and Westermann \cite{DBLP:conf/wea/EnglertRW06} by including the non-uniform cost models. They also proposed a new algorithm called Accelerated Threshold (Acc-T) which is a variant of TLC algorithm. In the experiments, they compared BW, MAP, RC, RR, TLC, and Acc-T algorithms against different combinations of input sequence types, cost functions, buffer sizes, and the number of colors. To evaluate the performance of the algorithms, they have defined Output/Input Switch Ratio, where a switch is defined as a color change that occurs between two consecutive items in the sequence. The result of the experiments on uniform cost scenarios indicated that BW achieves the highest reduction in switch ratio and on non-uniform cost scenarios Acc-T gives the best switch ratio across all combinations of input sequence types and cost functions for medium (50, 100, 500) and large buffer sizes.

\subsection{Our Results}

In this paper, we proved the minimum buffer length for the optimal solution to the reordering buffer management problem in the offline setting. With the assumption that color selection is always made when the buffer is full, selecting the most frequent color from the buffer given the smallest buffer size $k$ that satisfies either $o_1 < 2 \cdot \lceil \nicefrac{k}{\sigma} \rceil$ \textsc{OR} $o_2 < \lceil \nicefrac{k}{\sigma} \rceil$ guarantees the optimal solution, where $o_1$ and $o_2$ represent respectively the frequency of the most and the second most frequent colors in the input sequence $\mathcal{X}$, and $\sigma$ is the number of distinct colors appearing in $\mathcal{X}$. 

Additionally, we proposed an algorithm using the proved observation as a heuristic for the online setting of the problem. The observation on the minimum buffer length requires knowledge of the color size and frequency of the most and second most frequent color of the input sequence. However, in the online setting, we do not have this information. Our algorithm uses the limited information of buffer to decide whether buffer size is optimal for the sub-sequence of the input sequence that is inspected by buffer at the stage of selecting output color. When buffer size is not optimal, the algorithm uses a relaxation in which some items are allowed to be skipped in the first run. In a way, we are not forced to select the next item waiting in the queue, but we have the chance to observe more items. By means of this picking process, we refer to the algorithm as Picky algorithm.

Another significant outcome of the paper is the results of the first experimental evaluation setup for the online version of the reordering buffer management problem with uniform cost function on specific discrete distributions. We have generated data sets following different discrete probability distributions and tested the performance of algorithms Bounded Waste, Random Choice, and Picky. Out of 432 cases, Picky algorithm achieved 409 best cases that is $95\%$ of all cases.

In Section \ref{sec:analysis}, we gave proof of the minimum buffer length required for the optimal solution by providing a detailed case. In Section \ref{sec:algorithm}, we presented workflow and pseudo-code of Picky algorithm with illustrating the process on a case. In section \ref{sec:evaluation}, we gave the results of experimental evaluation per algorithm individually.

\section{The Analysis of Minimum Buffer Length for Optimal Solution}
\label{sec:analysis}

Let $\mathcal{X} = \langle x_1, x_2, \ldots x_n \rangle$ represent a sequence of $n$ items, where the color of $x_i$ is shown by $color(x_i)=c_i$ that is drawn from a color set $\mathcal{C}=\{c_1,c_2,\ldots c_\sigma\}$ of $\sigma$ distinct colors. The number of items in $\mathcal{X}$ with color $c_i$ is denoted by $o_i$, and since each item has only one associated color, $n=o_1 + o_2 + \ldots o_\sigma$. The items arrive one by one and  we have a buffer that can accommodate $k$ items. When the buffer is full, a color is selected as output color and all items with that color in the buffer are added to the output sequence creating vacancies for the newcomers. When the newcomers have the same color as the output color, they are added to the output sequence directly without any loss in the vacant positions of the buffer. The buffer becomes full when the items with the selected color have reached an end in the investigation range. Therefore, a new color is selected to continue the process. We aim to have the minimum number of color changes in the output sequence. 

Let $\mathcal{Y} = \langle (r_1,t_1),(r_2,t_2) \ldots (r_\ell,t_\ell) \rangle$ represent the output sequence where  $r_i \in \mathcal{C}$ is the selected color at $i$th step with $t_i$ appearances. For the optimal solution, we would like to find the minimum $\ell$ number such that  $n = t_1 + t_2 + \ldots t_\ell$. When we assume the buffer size $k$ is equal to the input size $n$, the optimal solution would be simply $\mathcal{Y} = \langle (c_1,o_1),(c_2,o_2) \ldots (c_\sigma,o_\sigma) \rangle$ with $\ell = \sigma$. However, the optimal solution can be guaranteed with a smaller buffer size, when the color selection strategy is always selecting the most frequent item from the buffer. Lemma \ref{lemma:mink} states that minimum buffer size.

\begin{lemma}
\label{lemma:mink}
Let $\mathcal{X} = \langle x_1, x_2, \ldots x_n \rangle$ represent a sequence of $n$ items and $\sigma$ is the number of distinct colors appearing in $\mathcal{X}$ where $o_1$ and $o_2$ represent respectively the frequency of the most and the second most frequent colors in $\mathcal{X}$. Assume the color selection is always made when the buffer is full, then with the strategy of selecting the most frequent color, the smallest buffer size $k$ that satisfies at least one of the following conditions guarantees the optimal solution. 

\begin{align} 
\label{lemma-eq-1}
2 \cdot \bigg\lceil \frac{k}{\sigma} \bigg\rceil &> o_1  \\ 
{\bigg\lceil} \frac{k}{\sigma} {\bigg\rceil} &> o_2 
\end{align}

\end{lemma}

\begin{proof}
The solution that is not optimal should include at least one more tuple than $\mathcal{Y} = \langle (c_1,o_1),(c_2,o_2) \ldots (c_\sigma,o_\sigma) \rangle$. This requires one of the colors to be represented by at least two different tuples. The scenario of the non-optimal solution can be summarized as follows. At some time, some color $c_i \in \mathcal{C}$ was the most frequent item in the buffer and selected. We will refer to this color $c_i$ as split color. When the buffer becomes full again after this selection, some other color is selected, and possibly followed by some other selections, after which $c_i$ again becomes dominant in the buffer. 
Notice that for a color to be selected in the size $k$ buffer, it should be observed at least $\lceil \nicefrac{k}{\sigma^\prime} \rceil$ times, where $\sigma^\prime$ is the number of observed colors in the buffer at the time of selection. 

Since $\sigma^\prime \leq \sigma$, and thus, $\lceil \nicefrac{k}{\sigma} \rceil  \leq \lceil \nicefrac{k}{\sigma^\prime} \rceil$,  we can say that the number of occurrences of any selected color should be larger than or equal to  $\lceil \nicefrac{k}{\sigma} \rceil$. This requires the split color $c_i$ to have an appearance count of at least $2 \cdot \lceil \nicefrac{k}{\sigma} \rceil$, since it is selected at least two times and the selection is always made when the buffer is full, that is $o_i \geq 2 \cdot \lceil \nicefrac{k}{\sigma} \rceil$.

Without loss of generality assume $o_1 \geq o_2 \geq \ldots \geq o_\sigma$, which means the most frequent color is $c_1$, and obviously $o_i \leq o_1$. Thus, $2 \cdot \lceil \nicefrac{k}{\sigma} \rceil  \leq o_{i} \leq o_1$.

For instance, the easiest non-optimal solution would have two ranges for the most frequent color, and instead of $(c_1,o_1)$ there will appear $(c_1, a_1)$ and $(c_1, b_1)$ tuples such that $a_1 + b_1 = o_1$, and these tuples are surely not subsequent. Splitting a color requires at least one other color to get most frequent in the buffer between the selections of the color $c_i$. Let us show this color with $c_j$, and its number of occurrences by $o_j$, where $o_j \leq o_1$.  To become the most frequent in the buffer at least once, $o_j \geq \lceil \nicefrac{k}{\sigma^\prime} \rceil \geq \lceil \nicefrac{k}{\sigma} \rceil $ should hold. It is important to notice that $j \neq i$, and the following two conditions should hold for distinct $i$ and $j$ values to be able to come up with a non-optimal solution with buffer size $k$.

\begin{align} 
\label{proof-eq-1}
2 \cdot \bigg\lceil \frac{k}{\sigma} \bigg\rceil &\leq o_i  \\ 
{\bigg\lceil} \frac{k}{\sigma} {\bigg\rceil} &\leq o_j 
\end{align}

\noindent The most possible indices for $i$ and $j$ are $i=1$ and $j=2$ since the possibility of the most frequent color existing in more than one tuple and the possibility that it is split by the second most frequent color is higher. Moreover, these choices are the least restrictive on the buffer size. Thus, the non-optimal solution can be obtained with a buffer size $k$ when the predicate $ P =  \big( 2 \cdot \lceil \nicefrac{k}{\sigma} \rceil  \leq o_{1} \big)$ \textsc{AND} $\big( \lceil \nicefrac{k}{\sigma} \rceil  \leq o_{2} \big)$ becomes true. Then, the buffer size $k$ which guarantees optimal solution should make $P$ false as $\neg  P = \big( 2 \cdot \lceil \nicefrac{k}{\sigma} \rceil  > o_{1} \big)$ \textsc{OR} $\big( \lceil \nicefrac{k}{\sigma} \rceil  > o_{2} \big)$, which finalizes the proof.
\qed
\end{proof}

\begin{remark}
The proof of the Lemma \ref{lemma:mink} strictly depends on the assumption that color selection is always made when the buffer is full. Although this assumption is strongly	restrictive, it comes directly from the definition of the problem. 

The main idea behind the Lemma \ref{lemma:mink} is to achieve a nearly optimal solution for any input size of $n$, neglecting the assumption of the buffer being always full at the end of the process. Almost anytime, at least a color selection will be made at the end of the process, when the number of items in the input sequence, which are not served by the algorithm yet, is less than buffer size $k$. Therefore, the buffer cannot be full anymore. However, the Lemma \ref{lemma:mink} guarantees the optimal solution for some part of any input sequence which is at least of size $n-k-1$. Let this part of the input sequence be denoted by $\mathcal{X}^\prime$. Assume a fixed buffer size $k$ is selected that satisfies the condition in Lemma \ref{lemma:mink} for the original input sequence $\mathcal{X}$. Let $o_1^\prime$ and $o_2^\prime$ be the most and second most frequent colors in $\mathcal{X}^\prime$, respectively. Since, $o_1^\prime \leq o_1$ and $o_2^\prime \leq o_2$, it is clear that the buffer size $k$ satisfies the condition of the Lemma \ref{lemma:mink} for the sequence $\mathcal{X}^\prime$, guaranteeing the optimal solution. Moreover, for any input sequence $\mathcal{X}$ with $n \geq k$, there is a sub-sequence of $\mathcal{X}$ such that every color selection is made when the buffer is full.

For example, let $\mathcal{X} = \langle 1,1,2,2,1,3,1,1,2,2,3,3,3,1,2,2,3,3,1 \rangle$ be an input sequence with $o_1 = 7$, $o_2 = 6$, $o_3 = 6$, and $\sigma = 3$, where the color of the items is represented as numbers. The smallest $k$ that satisfies $2 \cdot \lceil \nicefrac{k}{3} \rceil > 7 $ is 10 and the smallest $k$ that satisfies $\lceil \nicefrac{k}{3} \rceil > 6$ is 19. Thus, the smallest buffer size $k$ that satisfies the condition in Lemma \ref{lemma:mink} is 10. Consider the reordering buffer management strategy that selects the most frequent color in the buffer, evicts items of this color from the buffer and, adds them to the output sequence.

\begin{figure}
\scriptsize
\centering
\begin{tabular}{ |r c c c l|}
\hline
\multicolumn{5}{|c|}{\bfseries Stage 1} \\
Input sequence &  & Buffer &  & Output Sequence\\
\hline
3, 3, 3, 2, 2, 1, 1, 3, 1, 2, 2, 1, 1 & {\bfseries $\rightarrow$} & & {\bfseries $\rightarrow$} & \\
\hline
\hline

\multicolumn{5}{|c|}{\bfseries Stage 2} \\
Input sequence &  & Buffer &  & Output Sequence\\
\hline
1, 3, 3, 2, 2, 1, 3, 3, 3 & {\bfseries $\rightarrow$} & 2, 2, 1, 1, 3, 1, 2, 2, 1, 1 & {\bfseries $\rightarrow$} & \\
\hline
\hline

\multicolumn{5}{|c|}{\bfseries Stage 3} \\
Input sequence &  & Buffer &  & Output Sequence\\
\hline
1, 3, 3, 2, 2, 1, 3, 3, 3 & {\bfseries $\rightarrow$} & 2, 2, 3, 2, 2 & {\bfseries $\rightarrow$} & 1, 1, 1, 1, 1\\
\hline
\hline

\multicolumn{5}{|c|}{\bfseries Stage 4} \\
Input sequence &  & Buffer &  & Output Sequence\\
\hline
1, 3, 3, 2  & {\bfseries $\rightarrow$} & 2, 1, 3, 3, 3, 2, 2, 3, 2, 2 & {\bfseries $\rightarrow$} & 1, 1, 1, 1, 1\\
\hline
\hline

\multicolumn{5}{|c|}{\bfseries Stage 5} \\
Input sequence &  & Buffer &  & Output Sequence\\
\hline
1, 3, 3, 2  & {\bfseries $\rightarrow$} & 2, 3, 3 ,3, 2, 2, 3, 2, 2 & {\bfseries $\rightarrow$} & 1, 1, 1, 1, 1, 1\\
\hline
\hline

\multicolumn{5}{|c|}{\bfseries Stage 6} \\
Input sequence &  & Buffer &  & Output Sequence\\
\hline
1, 3, 3  & {\bfseries $\rightarrow$} & 2, 2, 3, 3 ,3, 2, 2, 3, 2, 2 & {\bfseries $\rightarrow$} & 1, 1, 1, 1, 1, 1\\
\hline
\hline

\multicolumn{5}{|c|}{\bfseries Stage 7} \\
Input sequence &  & Buffer &  & Output Sequence\\
\hline
1, 3, 3  & {\bfseries $\rightarrow$} & 3, 3, 3, 3 & {\bfseries $\rightarrow$} & 2, 2, 2, 2, 2, 2, 1, 1, 1, 1, 1, 1\\
\hline
\hline

\multicolumn{5}{|c|}{\bfseries Stage 8} \\
Input sequence &  & Buffer &  & Output Sequence\\
\hline
& {\bfseries $\rightarrow$} & 1, 3, 3, 3, 3, 3, 3 & {\bfseries $\rightarrow$} & 2, 2, 2, 2, 2, 2, 1, 1, 1, 1, 1, 1\\
\hline
\hline

\multicolumn{5}{|c|}{\bfseries Stage 9} \\
Input sequence &  & Buffer &  & Output Sequence\\
\hline
& {\bfseries $\rightarrow$} & 1 & {\bfseries $\rightarrow$} & 3, 3, 3, 3, 3, 3, 2, 2, 2, 2, 2, 2, 1\\

\hline
\hline
\multicolumn{5}{|c|}{\bfseries Stage 10} \\
Input sequence &  & Buffer &  & Output Sequence\\
\hline
& {\bfseries $\rightarrow$} & & {\bfseries $\rightarrow$} & 1, 3, 3, 3, 3, 3, 3, 2, 2, 2, 2, 2, 2\\
\hline
\end{tabular}
\caption{The process of the input sequence $\mathcal{X} = \langle 1,1,2,2,1,3,1,1,2,2,3,3,3,1,2,2,3,\\3,1\rangle$.}
\label{remark_example}
\end{figure}

\noindent The process of this sequence is given in Figure \ref{remark_example}. The process begins with an empty buffer and output sequence. In stage 2, the buffer is filled with items from the input sequence. Since the most frequent color in the buffer is 1, it is selected as the output color. Stage 3 demonstrates the process after items of output color are evicted from the buffer and added to the output sequence. In stage 4, the buffer is filled with the items from the input sequence. Since there is an item of the output color in the buffer, the algorithm does not perform a switch in the decision. In stage 5, one item of the output color is evicted from the buffer and added to the output sequence. In stage 6, the buffer is filled again and since the most frequent color in the buffer is 2, it is selected as output color. In stage 7, items of color 2 are evicted from the buffer and added to the output sequence. In stage 8, the remaining items in the input sequence are added to the buffer, and color 3 is selected as the output color since it is the most frequent color. In stage 9, items of color 3 are evicted from the buffer and added to the output sequence. In stage 10, the last item in the buffer is 1, thus it is selected as output color and added to the output sequence which ends the process.
After the input sequence is processed by the algorithm, the output sequence $\mathcal{Y} = \langle 1,1,1,1,1,1,2,2,2,2,2,2,3,3,3,3,3,3,1 \rangle$ is clearly not optimal. However, the optimal solution is achieved for the part of the input sequence $\mathcal{X}^\prime = \langle 1,1,2,2,1,3,1,1,2,2,3,3,3,1,2,2,3,3 \rangle$. 

\end{remark}

\section{Picky Algorithm}\label{sec:algorithm}

We propose a \textit{heuristic} algorithm, which we refer to as Picky algorithm, for the online version of the problem with uniform cost. The main idea of the Picky algorithm is to guarantee optimal solutions for sub-sequences of the input sequence. The algorithm uses the inequality (\ref{lemma-eq-1}) of Lemma \ref{lemma:mink} as a buffer management strategy. To implement the results of the Lemma \ref{lemma:mink} to the online setting, the algorithm only considers the knowledge of the items in the buffer at the stage of color selection. The algorithm tries to evict the most frequent color $o_1^\prime$ in the buffer as long as the buffer size is optimal which means that the inequality of $ 2 \cdot \lceil \nicefrac{k}{\sigma^\prime} \rceil  > o_1^\prime$ holds, where $\sigma^\prime$ is the color size of items that are present in the buffer at the stage of selection. When the inequality is violated, the algorithm evicts an item of the least frequent color from the buffer for the inequality to hold by diminishing the color size of the buffer. The evicted item is added to the end of the input sequence, and a new item from the input sequence is added to the buffer. This process is a relaxation to the original problem, in which some items are allowed to be skipped in the first run, and can be handled after the whole input sequence is processed for the first time. It can also be seen as a way to expand the investigation range of the buffer. By evicting some items the from buffer not to add to the output sequence but to add to the end of the input sequence, we are picking the desired items from the input sequence for that stage of the process. With this picking process, we guarantee that at the stage of the selection, buffer size $k$ is optimal for the sequence of items that are present in the buffer. Moreover, with the picking process, it is possible for some items that are not simultaneously in the buffer to be present in the buffer together in the second run. The workflow of the Picky algorithm is given in Appendix \ref{app_a}.

\begin{algorithm}[htbp]
 \tiny
 \SetAlgoLined
 \DontPrintSemicolon
 \KwIn{Input sequence: $\mathcal{X} = \langle x_1, x_2, \ldots x_n \rangle$, Buffer size: $k$}
 \KwOut{Output sequence: $\mathcal{Y} = \langle y_1, y_2, \ldots y_n \rangle$}
 output color $\leftarrow$ NULL\;
 $\mathcal{B}$ $\leftarrow$ $ \langle NULL \rangle$\;
 $\mathcal{Y}$ $\leftarrow$ $\langle NULL \rangle$\;
 
 \While{$ |\mathcal{X}| \neq$ 0}{
    exit-counter $\leftarrow$ 0\;
    
    \While{ $|\mathcal{B}|<$ k and $|\mathcal{X}| \neq$ 0}{
        fill $\mathcal{B}$ \;
        
    \If{ there are items with the  output-color in $\mathcal{B}$}{
        evict these items from $\mathcal{B}$ \Comment*[r]{Remove any item of the last selected output color from buffer}
        add them to $\mathcal{Y}$\;
        fill $\mathcal{B}$
        }
    $\sigma^\prime = ||\mathcal{B}||$\;
    $o_{1}^\prime =$ frequency of the most frequent color in $\mathcal{B}$\; 

    \While{$o_{1}^\prime \geq 2 \cdot \lceil \nicefrac{k}{\sigma^\prime} \rceil$ and exit-counter $<$ $|\mathcal{X}|$)}{
        evict an item of the least frequent color from $\mathcal{B}$ \Comment*[r]{If the condition in Lemma \ref{lemma:mink} is violated,}
        add evicted item to the end of $\mathcal{X}$\ \Comment*[r]{remove the least frequent color from buffer}
        exit-counter $++$\ \Comment*[r]{to diminish the color size}
        add the next item to $\mathcal{B}$\;
        }
        }
    \While{$|\mathcal{B}|$ $\neq$ 0}{
        output-color $\leftarrow$ most frequent color in $\mathcal{B}$\;
        evict items of output-color from $\mathcal{B}$\;
        add to $\mathcal{Y}$\;
        }
    }
 \While{$|\mathcal{B}|$ $\neq$ 0}{
    output-color $\leftarrow$ most frequent color in $\mathcal{B}$ \Comment*[r]{Remove remaining items starting from the most frequent color}
    evict items of output-color from $\mathcal{B}$\;
    add to $\mathcal{Y}$\;
    }
Return $\mathcal{Y}$\;
 \caption{Picky Algorithm}
 \label{alg:picky}  
\end{algorithm}

\noindent The pseudo-code of the Picky algorithm is given in Algorithm \ref{alg:picky}. The algorithm takes input sequence and buffer size $k$ as input, then creates an empty buffer of size $k$, a null output color, and an empty output sequence. If the input sequence is not empty, an exit counter is initiated. In the next step, the buffer is filled with items from the input sequence. After the buffer is filled, the algorithm checks if there is an item of the output color in the buffer. If there exists, the items are added to the output sequence and the buffer is filled again. When the buffer is full and clear of output color, the algorithm calculates $\sigma^\prime$ which is the color size of the buffer, and $o_1^\prime$ which is the frequency of most frequent color in the buffer to check whether $ 2 \cdot \lceil \nicefrac{k}{\sigma^\prime} \rceil  > o_1^\prime$ holds or not. If the inequality holds, it means buffer size $k$ is optimal for the sequence of items that are present in the buffer. Therefore, the algorithm simply selects the most frequent color as the output color, removes items of output color from the buffer, and adds them to the output sequence. 

However, if the inequality does not hold and the exit counter is less than the number of items remaining in the input sequence, then an item of the least frequent color is removed from the buffer, added to the end of the input sequence to be processed later and exit counter is incremented by one. After this step, the buffer is filled with the next item from the input sequence, and calculations are made again to check the inequality. This sub-process continues until the inequality holds or the exit counter becomes greater than the items remaining in the input sequence. After the algorithm finishes this sub-process, the buffer will be in one of the two positions, in the first one buffer size is optimal for the items in buffer, in the second one, there is no new item left in the input sequence to try for the inequality to hold. The algorithm treats both cases as choosing the most frequent color in the buffer as output color, removes items of output color from the buffer, and adds them to the output sequence. After the items are added to the output sequence, the algorithm returns to the first checking point and controls if there exists an item in the input sequence. If there exists, then the algorithm resets the exit counter and repeats the steps. When there is no item left in the input sequence, the algorithm checks if there exists an item in the buffer, removes them starting from the most frequent color, and adds them to the output sequence. The algorithm ends when both input sequence and buffer are empty, returns the output sequence as output.

\begin{figure}[htbp]
\scriptsize
\centering
\begin{tabular}{ |r c c c l|}
\hline
\multicolumn{5}{|c|}{\bfseries Stage 1} \\
Input sequence &  & Buffer &  & Output Sequence\\
\hline
3, 4, 4, 3, 2, 1, 1, 2, 2, 1, 1, 2, 2 & {\bfseries $\rightarrow$} & 3, 3, 3, 3, 2, 1, 3, 2, 1 & {\bfseries $\rightarrow$} & \\
\hline
\hline

\multicolumn{5}{|c|}{\bfseries Stage 2} \\
Input sequence &  & Buffer &  & Output Sequence\\
\hline
2, 3, 4, 4, 3, 2, 1, 1, 2 & {\bfseries $\rightarrow$} & 2, 1, 1, 2, 2, 2, 1, 2, 1 & {\bfseries $\rightarrow$} & 3, 3, 3, 3, 3 \\
\hline
\hline

\multicolumn{5}{|c|}{\bfseries Stage 3} \\
Input sequence &  & Buffer &  & Output Sequence\\
\hline
2, 3, 4, 4 & {\bfseries $\rightarrow$} & 3, 2, 1, 1, 2, 1, 1, 1, 1 & {\bfseries $\rightarrow$} & 2, 2, 2, 2, 2, 3, 3, 3, 3, 3\\
\hline
\hline

\multicolumn{5}{|c|}{\bfseries Stage 4} \\
Input sequence &  & Buffer &  & Output Sequence\\
\hline
2, 3  & {\bfseries $\rightarrow$} & 4, 4, 3, 1, 1, 1, 1, 1, 1 & {\bfseries $\rightarrow$} & 2, 2, 2, 2, 2, 2, 2, 3, 3, 3, 3, 3\\
\hline
\hline

\multicolumn{5}{|c|}{\bfseries Stage 5} \\
Input sequence &  & Buffer &  & Output Sequence\\
\hline
3, 2, 3  & {\bfseries $\rightarrow$} & 4, 4, 1, 1, 1, 1, 1, 1 & {\bfseries $\rightarrow$} & 2, 2, 2, 2, 2, 2, 2, 3, 3, 3, 3, 3\\
\hline
\hline

\multicolumn{5}{|c|}{\bfseries Stage 6} \\
Input sequence &  & Buffer &  & Output Sequence\\
\hline
3, 2,  & {\bfseries $\rightarrow$} & 3, 4, 4, 1, 1, 1, 1, 1, 1 & {\bfseries $\rightarrow$} & 2, 2, 2, 2, 2, 2, 2, 3, 3, 3, 3, 3\\
\hline
\hline

\multicolumn{5}{|c|}{\bfseries Stage 7} \\
Input sequence &  & Buffer &  & Output Sequence\\
\hline
3, 3  & {\bfseries $\rightarrow$} & 2, 4, 4, 1, 1, 1, 1, 1, 1 & {\bfseries $\rightarrow$} & 2, 2, 2, 2, 2, 2, 2, 3, 3, 3, 3, 3\\
\hline
\hline

\multicolumn{5}{|c|}{\bfseries Stage 8} \\
Input sequence &  & Buffer &  & Output Sequence\\
\hline
3 & {\bfseries $\rightarrow$} & 3, 4, 4 ,1, 1, 1, 1, 1, 1 & {\bfseries $\rightarrow$} & 2, 2, 2, 2, 2, 2, 2, 2, 3, 3, 3, 3, 3\\
\hline
\hline

\multicolumn{5}{|c|}{\bfseries Stage 9} \\
Input sequence &  & Buffer &  & Output Sequence\\
\hline
& {\bfseries $\rightarrow$} & 3, 3, 4, 4 & {\bfseries $\rightarrow$} & 1, 1, 1, 1, 1, 1, 2, 2, 2, 2, 2, 2, 2\\

\hline
\hline
\multicolumn{5}{|c|}{\bfseries Stage 10} \\
Input sequence &  & Buffer &  & Output Sequence\\
\hline
& {\bfseries $\rightarrow$} & 3, 3 & {\bfseries $\rightarrow$} & 4, 4, 1, 1, 1, 1, 1, 1, 2, 2, 2, 2, 2\\
\hline
\end{tabular}
\caption{The process of the input sequence $\mathcal{X} = \langle 1,2,3,1,2,3,3,3,3,2,2,1,1,2,2,1,1,2,\\3,4,4,3,2 \rangle$ with Picky algorithm.}
\label{picky_example}
\end{figure}

\noindent To demonstrate how picking process helps the algorithm to achieve better performance with smaller buffer size, consider input sequence $\mathcal{X} = \langle 1,2,3,1,2,3,3,\\
3,3,2,2,1,1,2,2,1,1,2,3,4,4,3,2\rangle$ with $o_1 = 8$, $o_2 = 7$, and $\sigma = 4$, where the color of the items is represented as numbers. The smallest $k$ that satisfies $2 \cdot \lceil \nicefrac{k}{4} \rceil > 8 $ is 17 and the smallest $k$ that satisfies $\lceil \nicefrac{k}{4} \rceil > 7$ is 29. Thus, the smallest buffer size $k$ that satisfies the condition in Lemma \ref{lemma:mink} is 17. Consider the process of input sequence $\mathcal{X}$ by Picky algorithm with buffer size $k=9$. 

The process is given in Figure \ref{picky_example}. In stage 1, buffer is filled, $\sigma^\prime = 3$, $o_{1}^\prime = 5$, and  $5 < 2 \cdot \lceil \nicefrac{9}{3} \rceil$. Therefore, the algorithm selects the most frequent color in the buffer as output color, evicts items of color $3$ from the buffer, and adds them to output sequence $\mathcal{Y}$. In stage 2, the vacancies in the buffer are filled with waiting items from $\mathcal{X}$, $\sigma^\prime = 2$, $o_{1}^\prime = 5$, and  $5 < 2 \cdot \lceil \nicefrac{9}{2} \rceil$. Hence, the algorithm selects the most frequent color in the buffer as output color, evicts items of color $2$ from the buffer, and adds them to output sequence $\mathcal{Y}$. In stage 3, vacancies in the buffer are filled with waiting items in $\mathcal{X}$ and since there exist items of output color in the buffer, the algorithm does not perform a switch in the decision and evicts the items of output color from the buffer. In stage 4, vacancies in the buffer are filled with waiting items in $\mathcal{X}$, $\sigma^\prime = 3$, $o_{1}^\prime = 6$, and  $6 \geq 2 \cdot \lceil \nicefrac{9}{3} \rceil$. Thus, the algorithm decides to remove an item of the least frequent color from the buffer. The least frequent color in buffer at stage 4 is $3$, thus in stage 5 an item of $3$ is evicted from the buffer and added to the end of $\mathcal{X}$. In stage 6 buffer is filled, $\sigma^\prime = 3$, $o_{1}^\prime = 6$, and  $6 \geq 2 \cdot \lceil \nicefrac{9}{3} \rceil$. The algorithm again decides to remove an item of the least frequent color from the buffer that is $3$. In stage 7, an item of color $3$ is removed from the buffer, added to the end of $\mathcal{X}$, and vacancy in the buffer is filled with the next item from $\mathcal{X}$ that is $2$. Since $2$ is the last selected color as output color, the algorithm does not make a switch in the decision and removes it from buffer to add $\mathcal{Y}$. In stage 8, after the buffer is filled, $\sigma^\prime = 3$, $o_{1}^\prime = 6$, and  $6 \geq 2 \cdot \lceil \nicefrac{9}{3} \rceil$. Thus, the algorithm decides to remove the least frequent color from the buffer. However after color $3$ is removed from the buffer, it is added again to the buffer to fill. The algorithm exits from the picking process, when the exit counter becomes greater than the number of remaining items in $\mathcal{X}$. After this step, all remaining items of $\mathcal{X}$ are added to the buffer, and the algorithm starts removing all items in the buffer starting from the most frequent color and adds to the output sequence. The process ends when both input sequence and buffer is empty and returns output sequence $\mathcal{Y} = \langle 3, 3, 3, 3, 3, 2, 2, 2, 2, 2, 2, 2, 2, 1, 1, 1, 1, 1, 1, 4, 4, 3, 3 \rangle$. If the same input sequence was processed by only selecting the most frequent color in buffer, the output sequence would have become $\mathcal{Y^\prime} = \langle 3, 3, 3, 3, 3, 2, 2, 2, 2, 2, 2, 2, 1, 1, 1, 1, 1, 1, 3, 3, 4,\\4, 2 \rangle$, which is less optimal than $\mathcal{Y}$. The picking process in stage 7, expanded the buffer's investigation area and enabled items of color 2 to be presented as one tuple in the output sequence.

\section{Experimental Evaluation}\label{sec:evaluation}
\subsection{Setup}

We have generated data sets following discrete distributions; Uniform, Binomial, Negative Binomial, Geometric, Poisson, and Zipf distributions with different parameters when possible. Poisson, Binomial, Geometric, and Negative Binomial distributions are used mainly in statistical modeling. Zipf distribution originates from Zipf's law which is an empirical law formulated using mathematical statistics that occurs in many rankings of human-created systems \cite{piantadosi2014zipf}. When creating the data sets, truncated distributions are used for Negative Binomial, Geometric, Poisson, and Zipf distributions. With different parameters, we have used 16 discrete distributions. A data set has 3 parameters as input size, color size, and distribution. We have used input sizes 1.000, 5.000, and 10.000. For color size, we have used $1\%$, $2\%$, and $5\%$ of input size, for each input size. Thus, we have generated 144 ($3\cdot3\cdot16$) different data sets. We created 50 random input sequences for each data set to average the results.

For the online version of the problem with uniform cost, the algorithms proposed are FIFO, LRU, LCF, BW, LCC, RC, and RR. FIFO, LRU, and LCF are proved to be unsuitable for the problem by Räcke, Sohler, and Westermann \cite{DBLP:conf/esa/RackeSW02}. RC and RR are two variants of MAP strategy and as shown by Englert, Röglin, and Westermann \cite{DBLP:conf/wea/EnglertRW06} they perform similarly. LCC algorithm uses the same concept as BW algorithm and very similar in the uniform setting. Therefore, we have chosen Random Choice as an indicator of a random strategy and Bounded Waste as an indicator of a deterministic strategy to compare with Picky strategy.

We have implemented algorithms to the data sets with buffer sizes equal to $1\%$, $2\%$, and $5\%$ of input sizes, for each data set. Thus, an experiment has 4 parameters as input size, color size, distribution, and buffer size that results in 432($3\cdot144$) experiments per algorithm, and 27 experiments per distribution and algorithm. Every result of each experiment is represented as the average result of 50 random input sequences related to that data set.

To evaluate the performance, we have used Output/Input Switch Ratio (\ref{switch-ratio}) where a switch is defined as a color change that occurs between two consecutive items in the sequence. A smaller Output/Input Switch Ratio implies better performance. 

\begin{align}\label{switch-ratio}
    \mbox{Output/Input Switch Ratio} = \frac{\mbox{Number of switches in output sequence}}{\mbox{Number of switches in input sequence}}
\end{align}

Because of the picking process in Picky algorithm, the algorithm can run some parts of the input more than once. Therefore, for Picky algorithm, we have also computed the Excess Run(\ref{excess-run}). In all experiments, we have counted how many items are removed from the buffer to be added to the end of the input sequence and divided this number by the input size. 
\begin{align}\label{excess-run}
    \mbox{Excess Run} = \frac{\mbox{Number of items skipped in the buffer}}{\mbox{Size of input sequence}}
\end{align}

We have used Python 3 to generate the data sets and implement the algorithms. Related research data and codes are hosted on a repository at GitHub, \url{https://github.com/gozdefiliz/Reordering-Buffer-Management}. 

\subsection{Results}

Average results per distribution and algorithm are given in Table \ref{tab:cumul_results}. According to average results per distribution, Picky algorithm achieved the best and Random Choice algorithm achieved the worst performance. Out of 432 cases, Picky algorithm achieved 409 best cases that is $95\%$ of all cases. Considering all experiments, the average excess run for Picky algorithm is 1,66. All algorithms achieved their best performance on geometric distribution with $p = 0,7$, and their worst performance on uniform distribution.  

Performance and excess run results can be read from Table \ref{tab:cumul_results} for each distribution. For each distribution, a row in the table represents the average results of the algorithms in 27 experiments. For instance, for uniform distribution, Picky algorithm is best with Output/Input Switch Ratio 0,381 followed by BW algorithm with 0,415. Out of the 27 different experiments, Picky algorithm performed best in all cases. The average excess run for Picky algorithm is 0,55.

\begin{table}[ht]
    \scriptsize
    \caption{Average results per distribution and algorithm.}
    \label{tab:cumul_results}
    \centering
    \begin{tabular}{|>{\centering\arraybackslash}m{0.2\linewidth}|>{\centering\arraybackslash}m{0.1\linewidth}|>{\centering\arraybackslash}m{0.1\linewidth}|>{\centering\arraybackslash}m{0.1\linewidth}|>{\centering\arraybackslash}m{0.1\linewidth}|>{\centering\arraybackslash}m{0.1\linewidth}|>{\centering\arraybackslash}m{0.1\linewidth}|>{\centering\arraybackslash}m{0.1\linewidth}|}
        \hline
        {} & \multicolumn{2}{c|}{Bounded Waste} & \multicolumn{2}{c}{Random Choice} & \multicolumn{3}{|c|}{Picky}\\
        \hline
        Distribution&Average Switch Ratio&Number of Best Cases&Average Switch Ratio&Number of Best Cases&Average Switch Ratio&Number of Best Cases&Average Excess Run\\ 
        \hline
        \multicolumn{1}{|r|}{Uniform}  & 0,415 & 0 & 0,452 & 0 & 0,381 & 27 & 0,55 \\ 
        \multicolumn{1}{|r|}{Binomial $p=0,3$} & 0,149 & 2 & 0,172 & 0 & 0,111 & 25 & 2,54 \\
        \multicolumn{1}{|r|}{Binomial $p=0,5$} & 0,160 & 1 & 0,185 & 0 & 0,118 & 26 & 2,54 \\
        \multicolumn{1}{|r|}{Binomial $p=0,7$} & 0,149 & 2 & 0,172 & 0 & 0,115 & 25 & 2,5 \\
        \multicolumn{1}{|r|}{N.Binomial $p=0,3$} & 0,097 & 0 & 0,111 & 0 & 0,056 & 27 & 1,85 \\
        \multicolumn{1}{|r|}{N.Binomial $p=0,5$} & 0,1978 & 2 & 0,227 & 0 & 0,16 & 25 & 2,39 \\
        \multicolumn{1}{|r|}{N.Binomial $p=0,7$} & 0,217 & 3 & 0,248 & 0 & 0,178 & 24 & 2,43 \\
        \multicolumn{1}{|r|}{Geometric $p=0,3$} & 0,106 & 3 & 0,121 & 0 & 0,078 & 24 & 1,80 \\
        \multicolumn{1}{|r|}{Geometric $p=0,5$} & 0,069 & 0 & 0,078 & 0 & 0,036 & 27 & 0,86 \\
        \multicolumn{1}{|r|}{Geometric $p=0,7$} & 0,049 & 0 & 0,053 & 0 & 0,03 & 27 & 0,33 \\
        \multicolumn{1}{|r|}{Poisson $m=1$} & 0,054 & 0 & 0,060 & 0 & 0,037 & 27 & 0,78 \\
        \multicolumn{1}{|r|}{Poisson $m=2$} & 0,074 & 0 & 0,084 & 0 & 0,051 & 27 & 1,27 \\
        \multicolumn{1}{|r|}{Poisson $m=3$} & 0,089 & 0 & 0,102 & 0 & 0,066 & 27 & 1,53 \\
        \multicolumn{1}{|r|}{Zipf $a=1,1$} & 0,272 & 5 & 0,299 & 0 & 0,234 & 22 & 2,45 \\
        \multicolumn{1}{|r|}{Zipf $a=1,5$} & 0,188 & 3 & 0,208 & 0 & 0,150 & 24 & 1,86 \\
        \multicolumn{1}{|r|}{Zipf $a=2$} & 0,123 & 2 & 0,135 & 0 & 0,091 & 25 & 0,93 \\
        \hline
         
    \end{tabular}
    
    \label{tab:my_label}
\end{table}

\noindent The most important variables for the performance of algorithms are buffer size and color size. From the results, it can be stated that, generally, the performance of the algorithms is inversely proportional to the color size and directly proportional to the buffer size. However, Picky algorithm is more sensitive to the changes in the variables such that it benefits more from an increase in buffer size and is further affected by an increase in color size. An exception to this observation is Poisson distribution. In Poisson distribution, changes in both color and buffer size do not affect the performance of the algorithms in a pattern. Furthermore, in Geometric distribution, the effect of the color size decreases as the input size increases. Picky algorithm's excess run is generally directly proportional to input size and buffer size. However, the effect of the color size does not occur to be consistent. Detailed graphs for each distribution are provided in Appendix \ref{app_b}.

\section{Conclusion}
In this paper, we have proved the minimum buffer length for the optimal solution to the reordering buffer management problem in the offline setting. Furthermore, using the observations made on the proof, we have proposed a heuristic algorithm, which we refer to as Picky algorithm, for the online setting of the problem. Moreover, we have made an experimental evaluation of 3 algorithms for the problem in different discrete distributions and represented the results. Out of 432 cases, Picky algorithm achieved 409 best cases that is $95\%$ of all cases.

To present directions for future work on Picky algorithm, proving the competitive ratio of the algorithm will be a contribution. Furthermore, the algorithm can be improved to support the non-uniform cost function. For the experiments in discrete distributions, the experiments can be improved to include other algorithms, especially algorithms that support non-uniform cost function.

\bibliographystyle{splncs04}
\bibliography{references}

\newpage
\begin{subappendices}
\renewcommand{\thesection}{\Alph{section}}
\section{}
\label{app_a}
\begin{figure} [ht]
\centering
\includegraphics[width=1.0\textwidth]{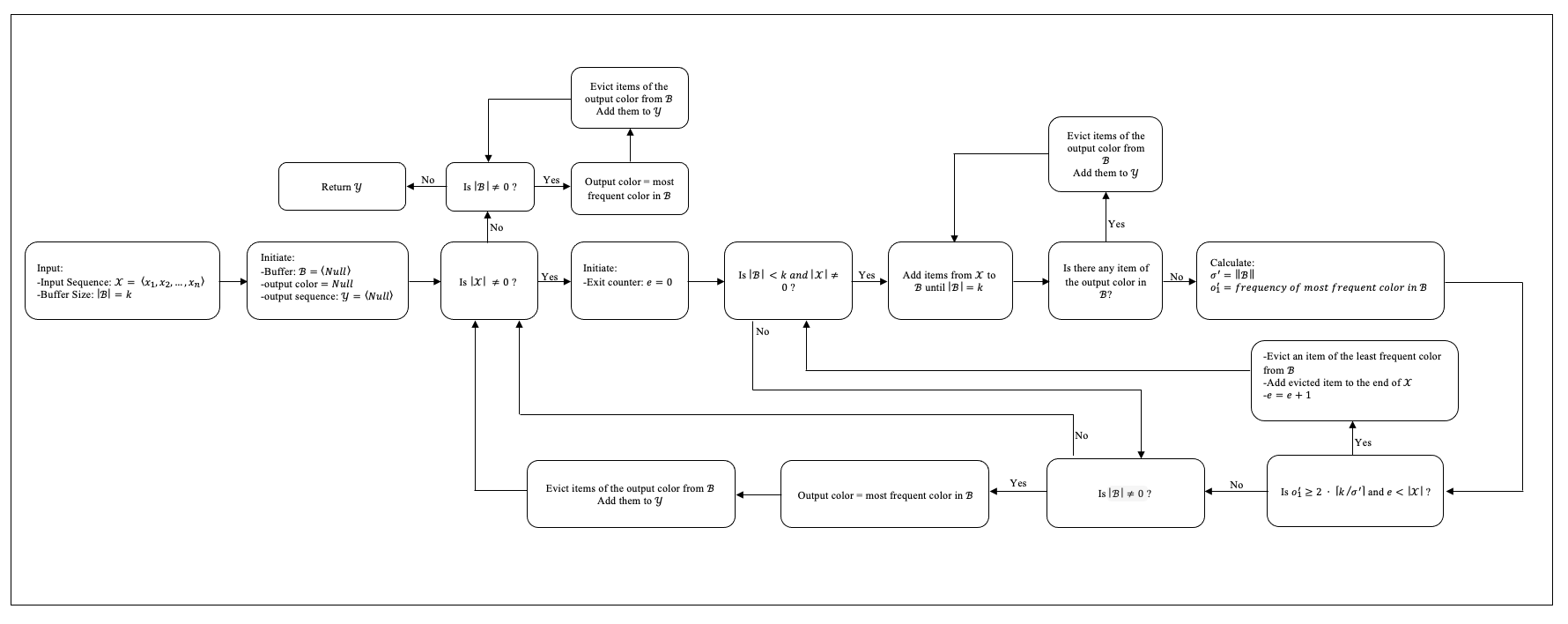}
\caption{The workflow of Picky algorithm.}\label{picky_diag} 
\end{figure}

\section{}
\label{app_b}

\begin{figure}[htbp]
\centering
\includegraphics[width=1.0\textwidth]{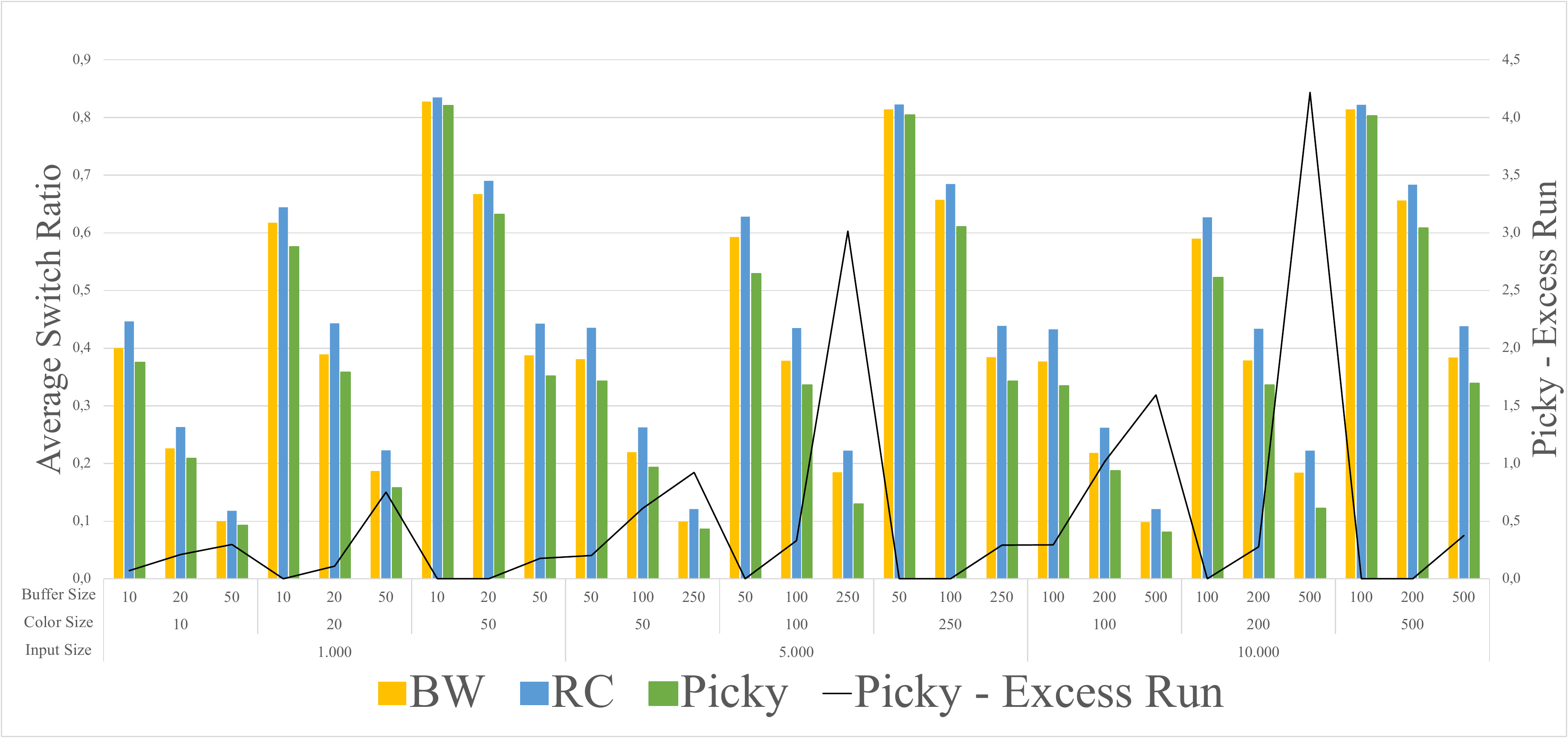}
\caption{Output/Input Switch Ratios for Uniform Distribution}
\end{figure}

\begin{figure}[htbp]
\centering
\includegraphics[width=1.0\textwidth]{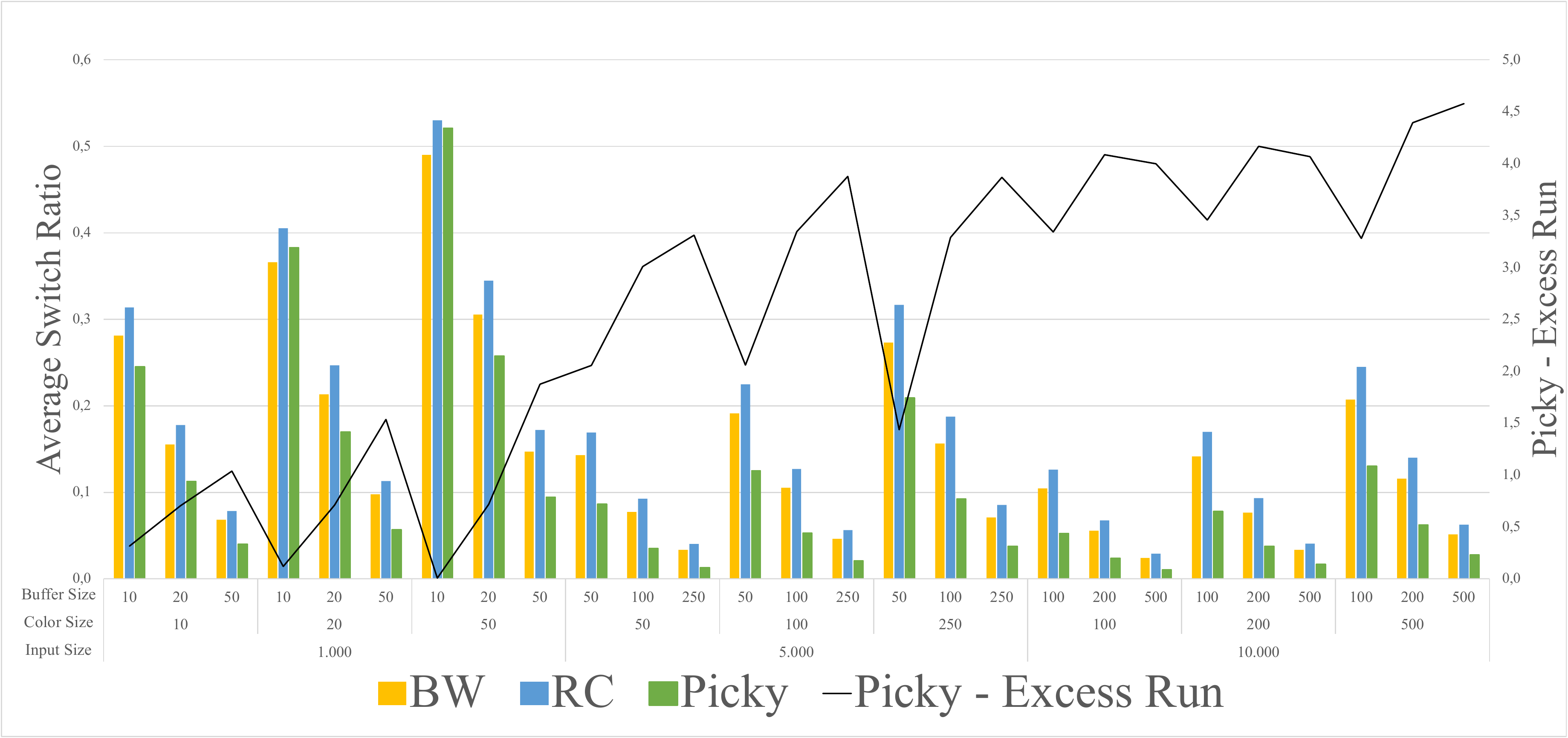}
\caption{Output/Input Switch Ratios for Binomial Distribution with $p=0,3$}
\end{figure}

\begin{figure}[htbp]
\centering
\includegraphics[width=1.0\textwidth]{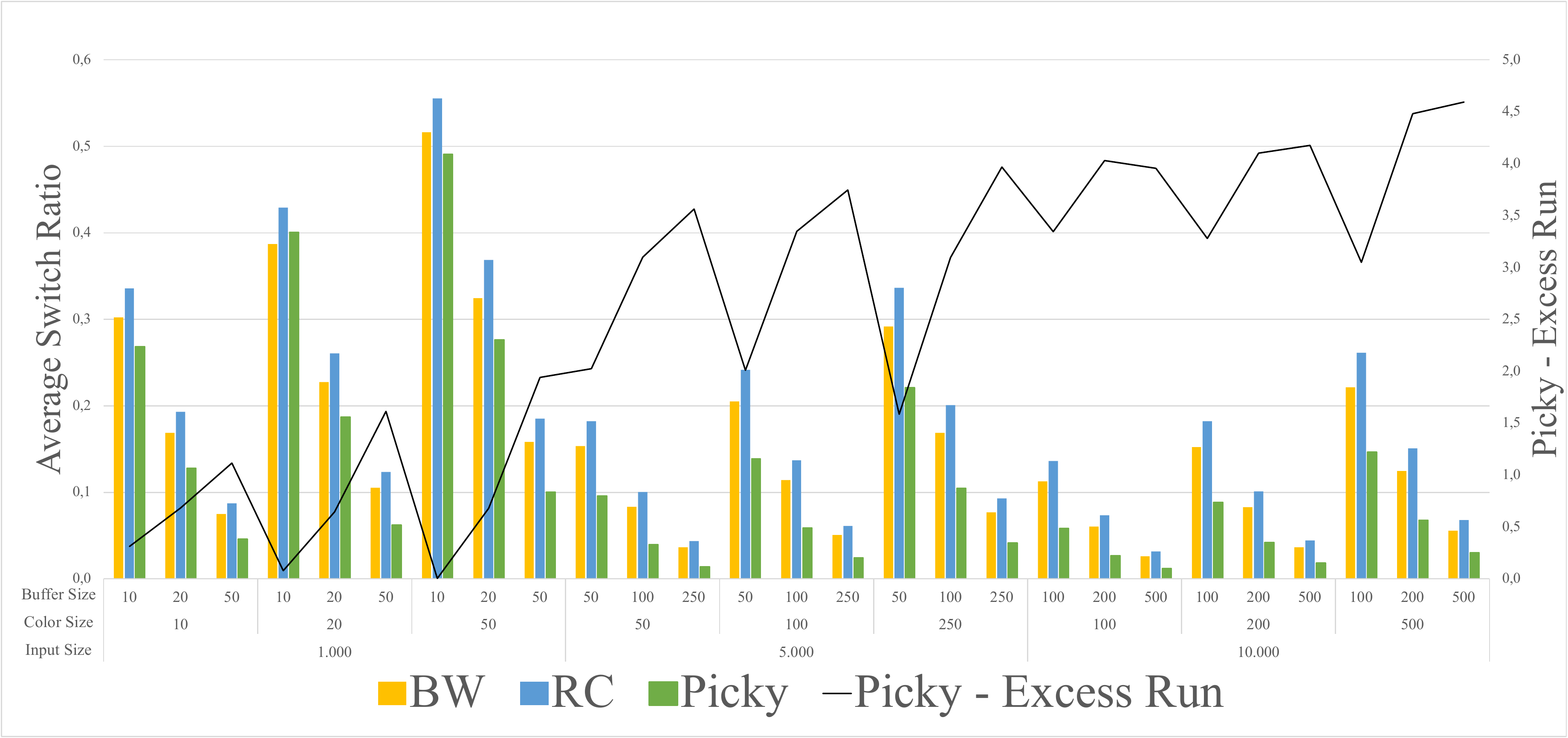}
\caption{Output/Input Switch Ratios for Binomial Distribution with $p=0,5$}
\end{figure}

\begin{figure}[htbp]
\centering
\includegraphics[width=1.0\textwidth]{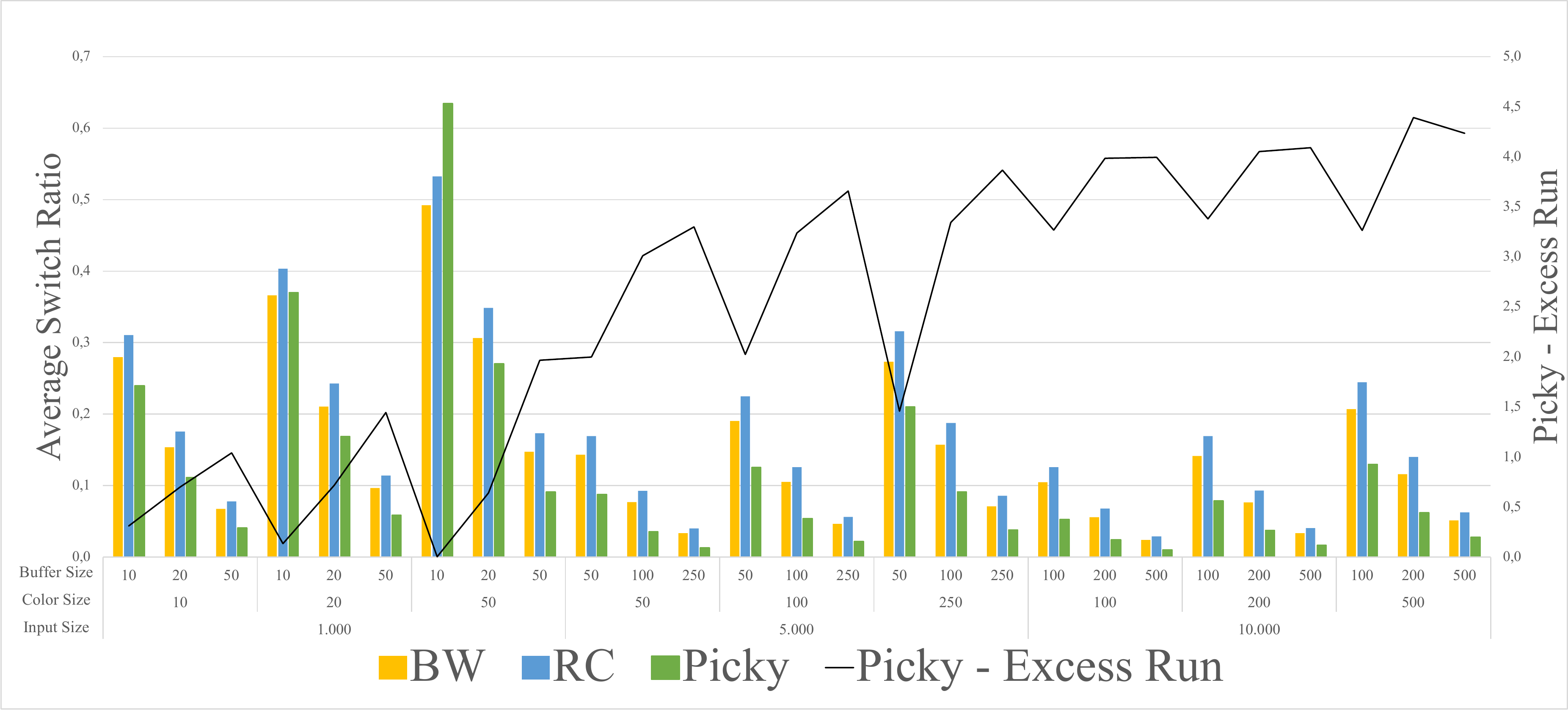}
\caption{Output/Input Switch Ratios for Binomial Distribution with $p=0,7$}
\end{figure}

\begin{figure}[htbp]
\centering
\includegraphics[width=1.0\textwidth]{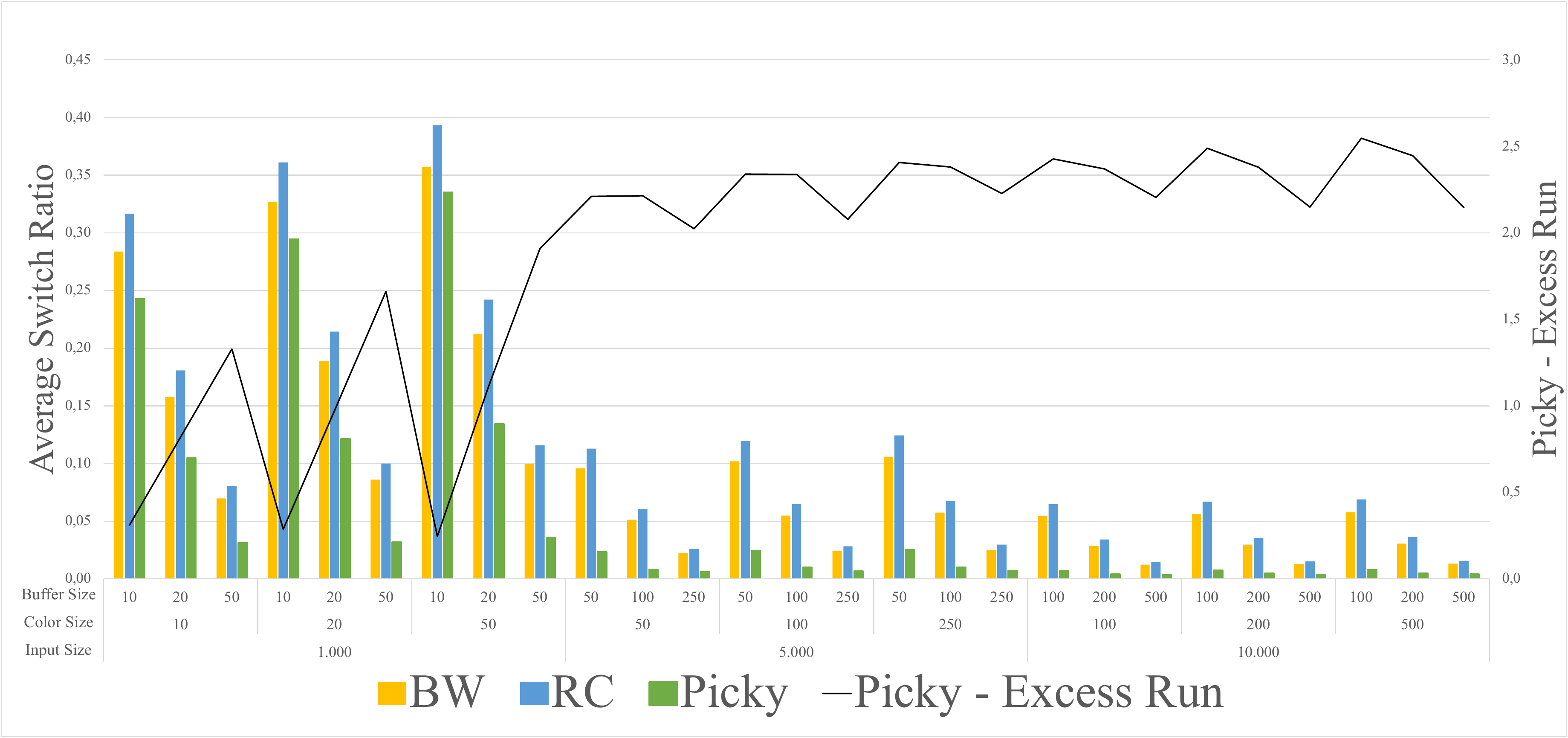}
\caption{Output/Input Switch Ratios for Negative Binomial Distribution with $p=0,3$}
\end{figure}

\begin{figure}[htbp]
\centering
\includegraphics[width=1.0\textwidth]{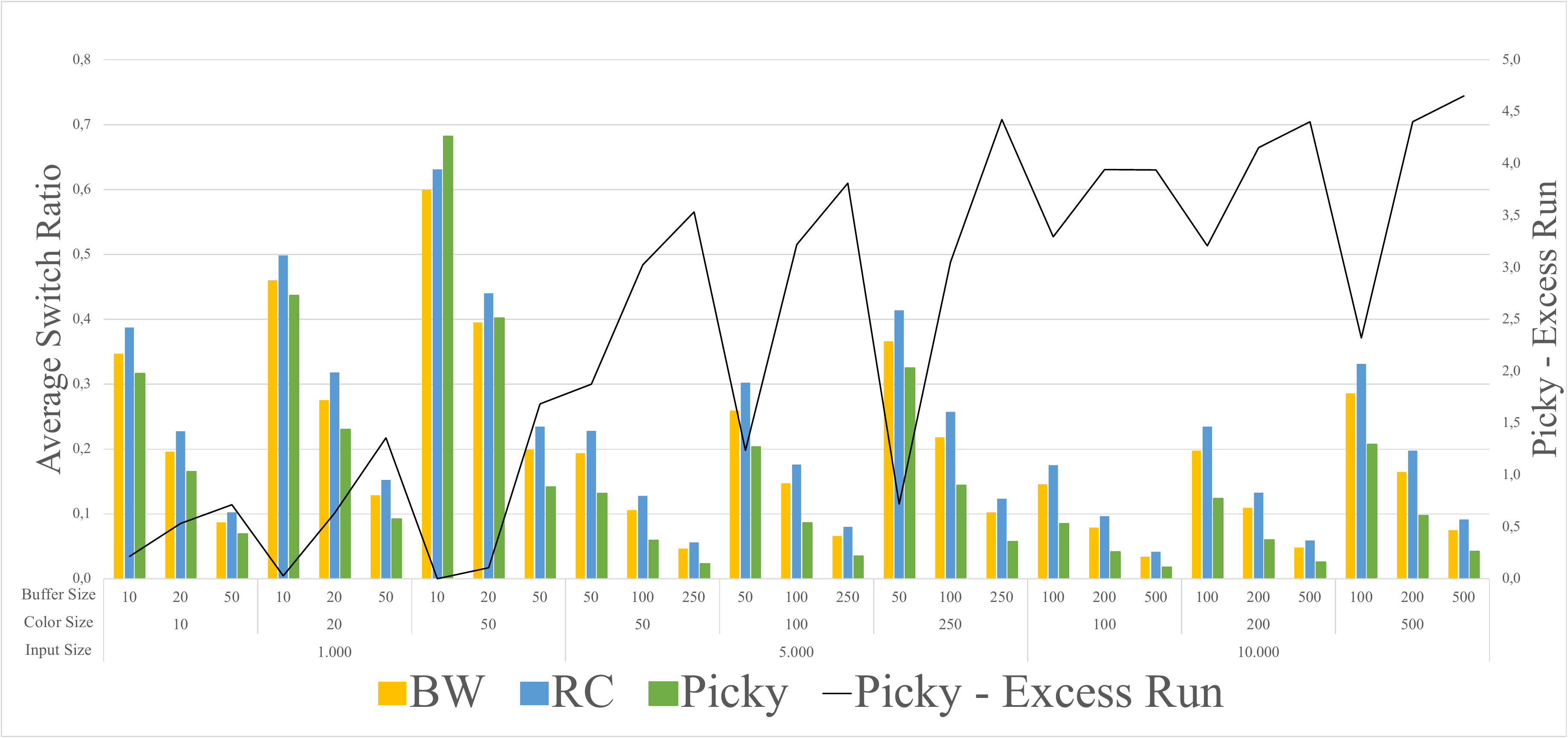}
\caption{Output/Input Switch Ratios for Negative Binomial Distribution with $p=0,5$}
\end{figure}

\begin{figure}[htbp]
\centering
\includegraphics[width=1.0\textwidth]{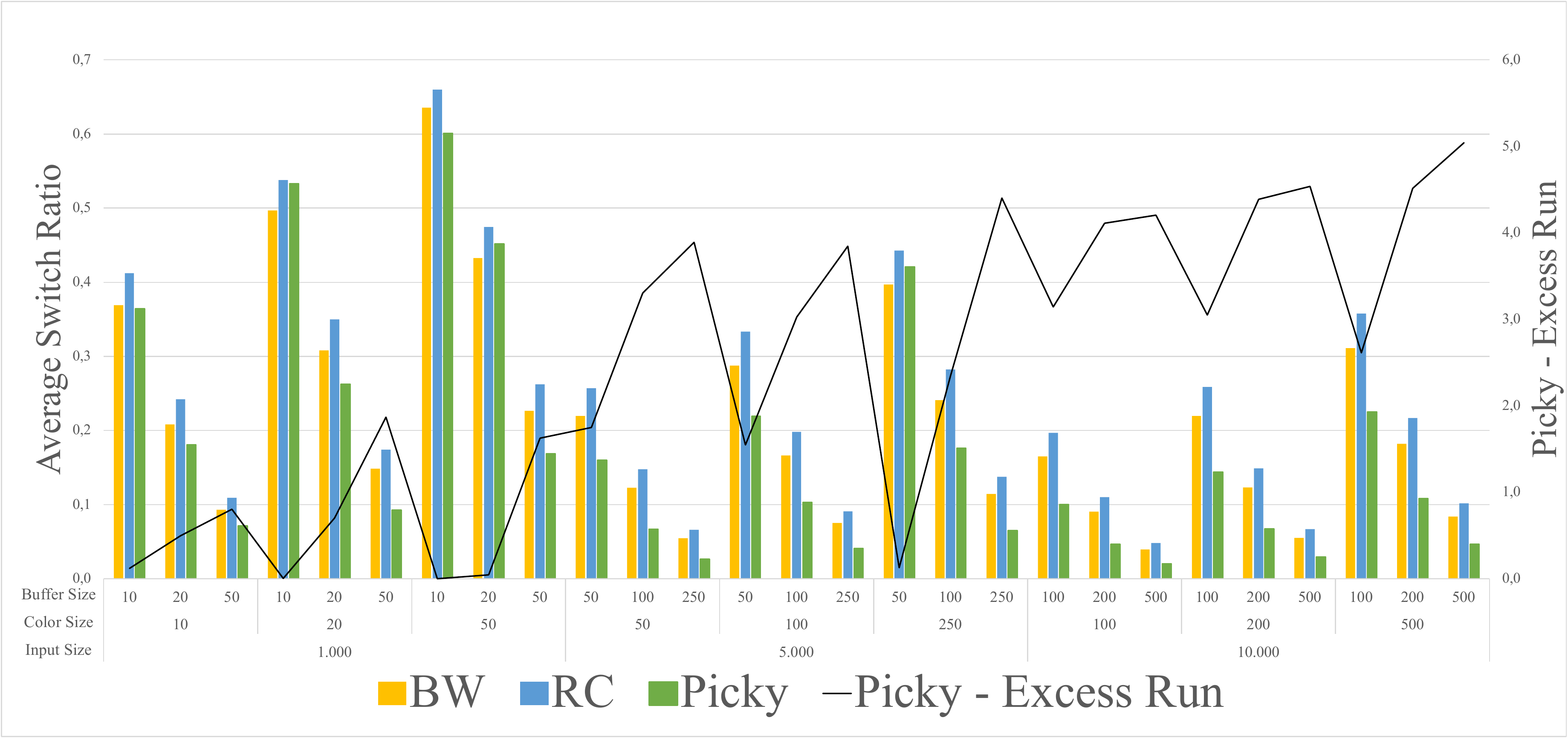}
\caption{Output/Input Switch Ratios for Negative Binomial Distribution with $p=0,7$}
\end{figure}

\begin{figure}[htbp]
\centering
\includegraphics[width=1.0\textwidth]{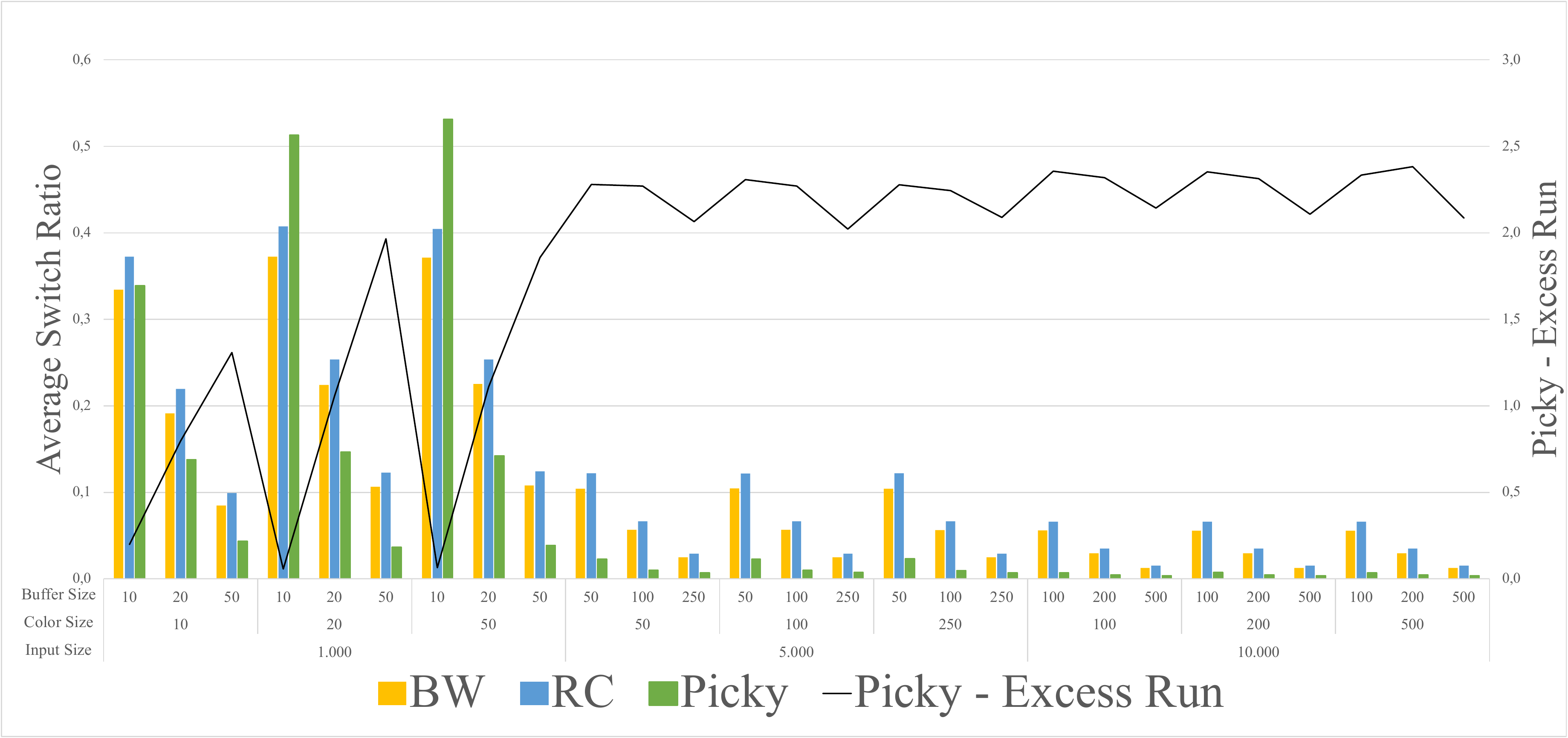}
\caption{Output/Input Switch Ratios for Geometric Distribution with $p=0,3$}
\end{figure}

\begin{figure}[htbp]
\centering
\includegraphics[width=1.0\textwidth]{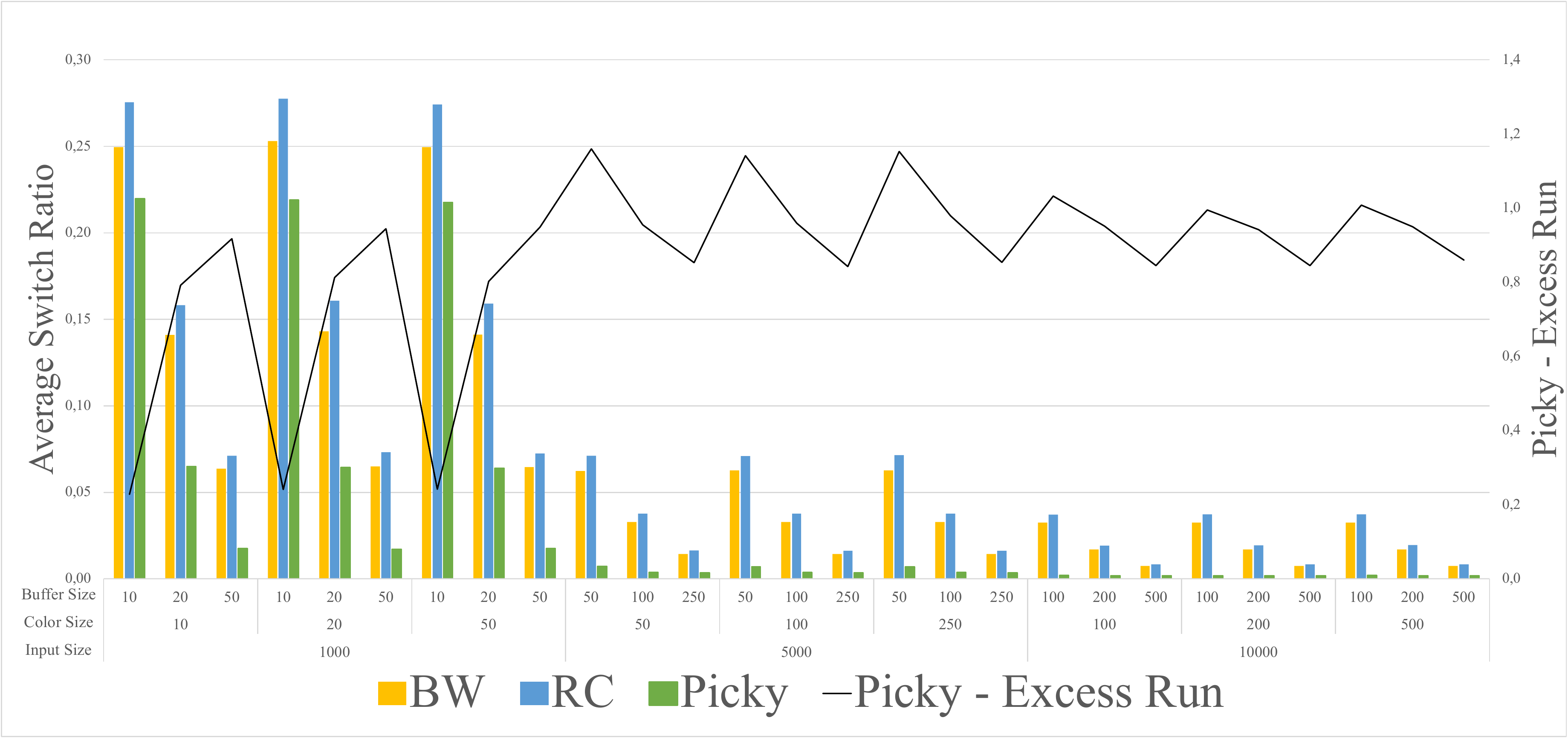}
\caption{Output/Input Switch Ratios for Geometric Distribution with $p=0,5$}
\end{figure}

\begin{figure}[htbp]
\centering
\includegraphics[width=1.0\textwidth]{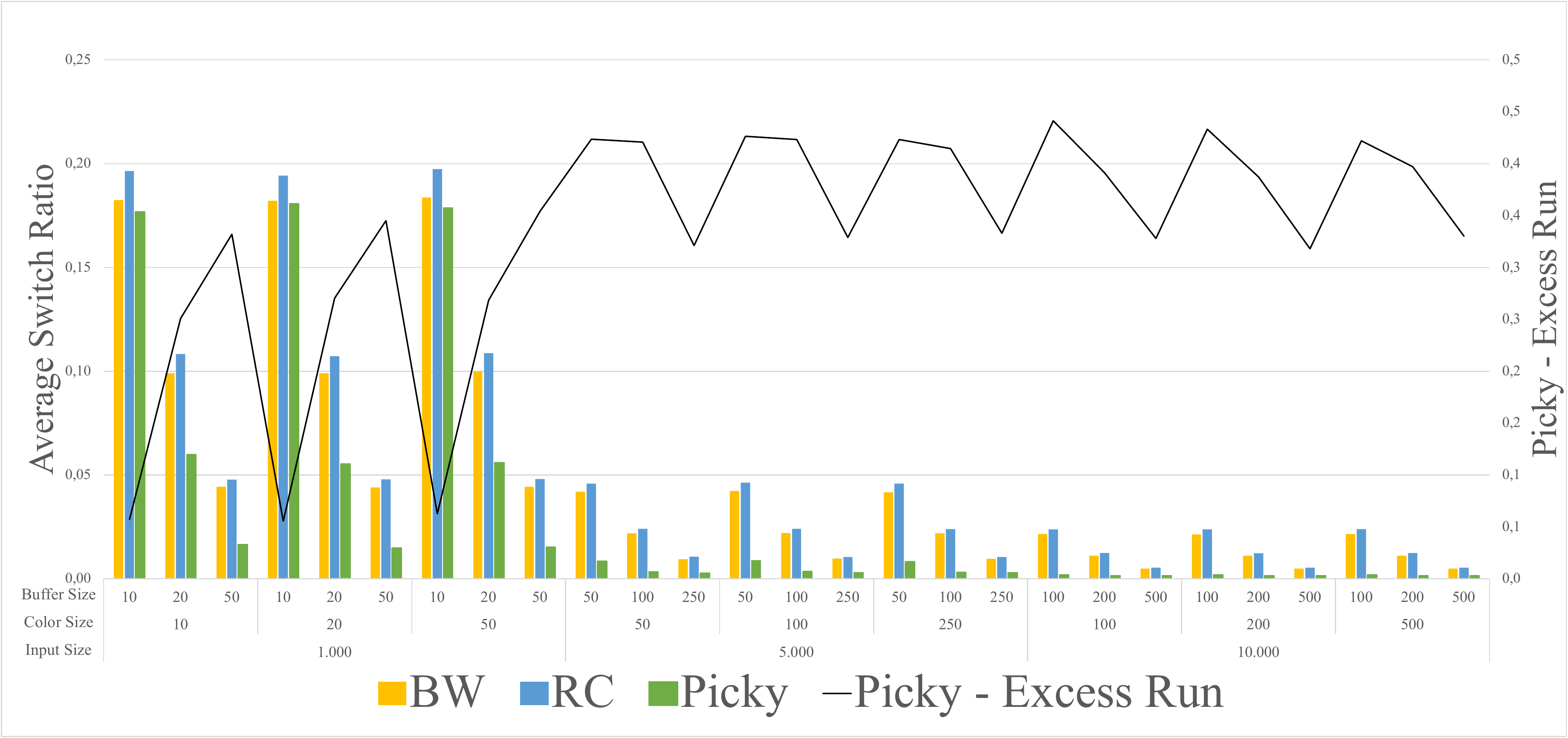}
\caption{Output/Input Switch Ratios for Geometric Distribution with $p=0,7$}
\end{figure}

\begin{figure}[htbp]
\centering
\includegraphics[width=1.0\textwidth]{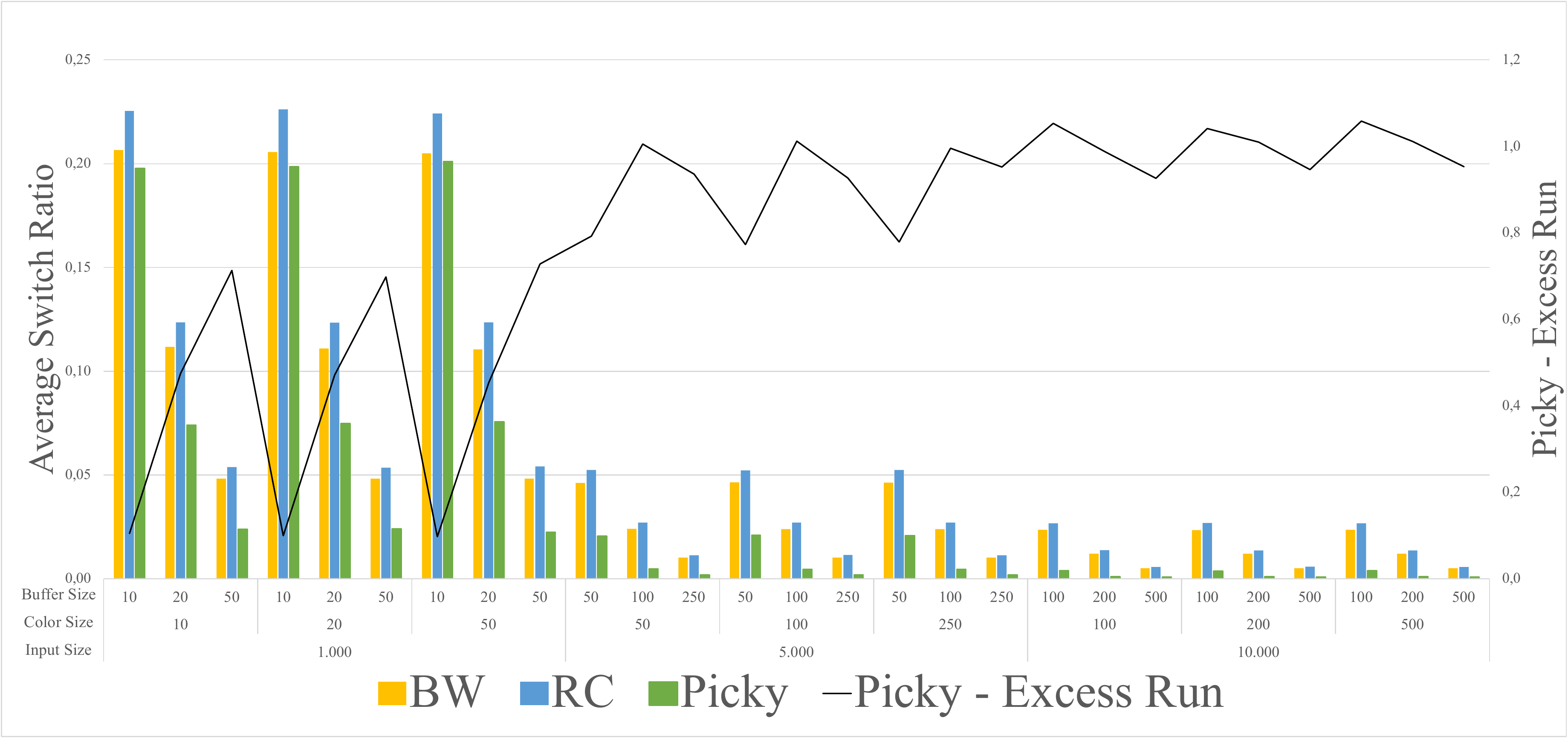}
\caption{Output/Input Switch Ratios for Poisson Distribution with $m=1$}
\end{figure}

\begin{figure}[htbp]
\centering
\includegraphics[width=1.0\textwidth]{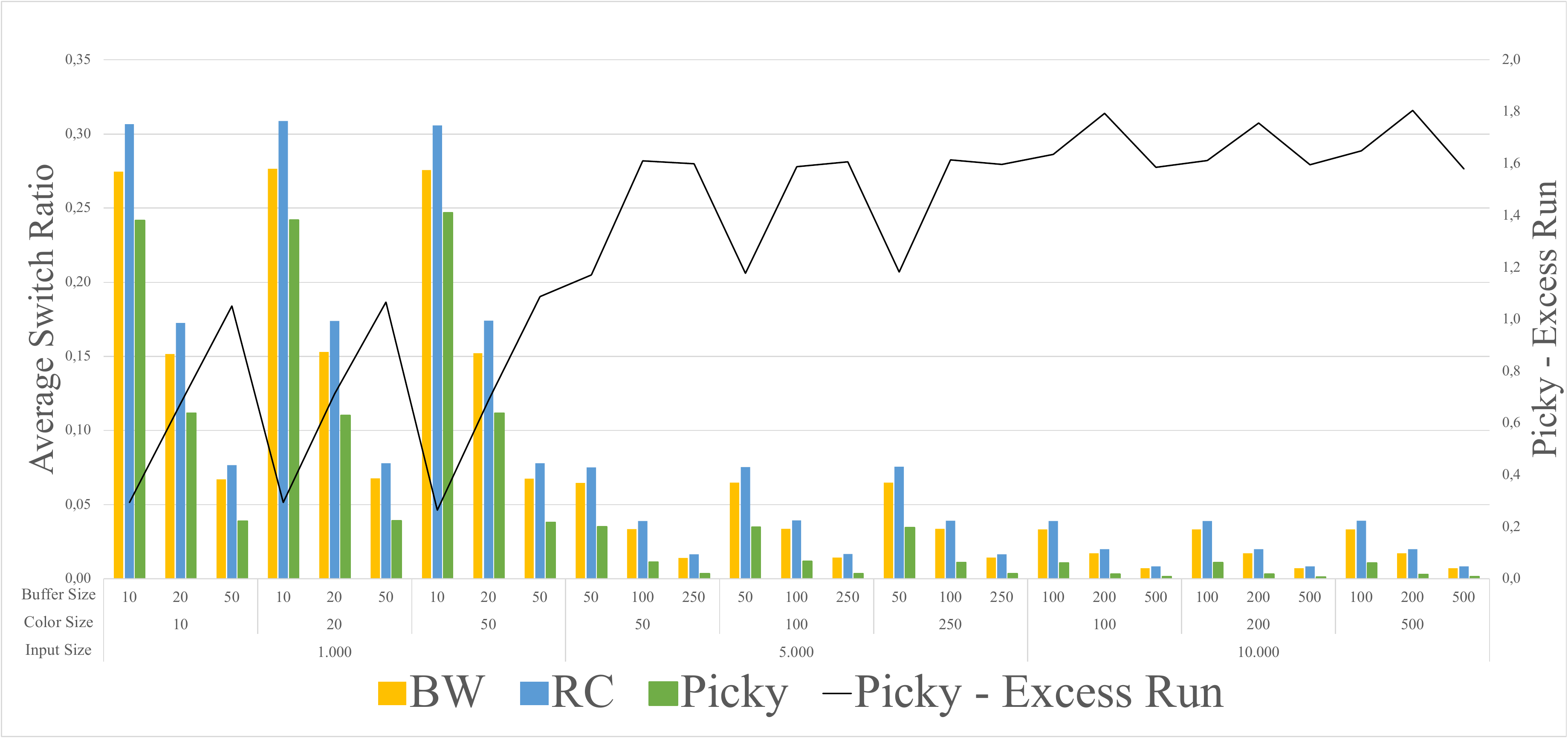}
\caption{Output/Input Switch Ratios for Poisson Distribution with $m=2$}
\end{figure}

\begin{figure}[htbp]
\centering
\includegraphics[width=1.0\textwidth]{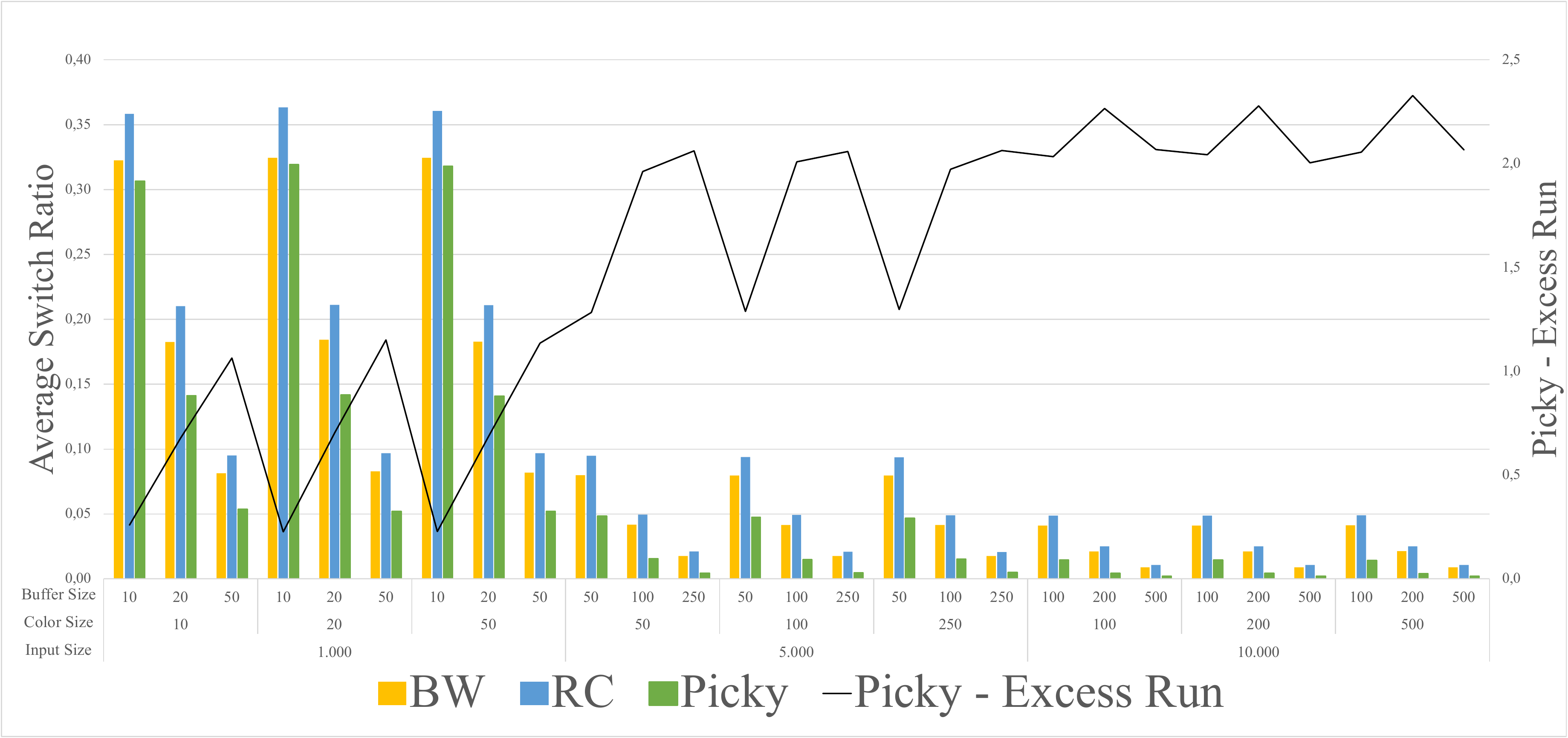}
\caption{Output/Input Switch Ratios for Poisson Distribution with $m=3$}
\end{figure}

\begin{figure}[htbp]
\centering
\includegraphics[width=1.0\textwidth]{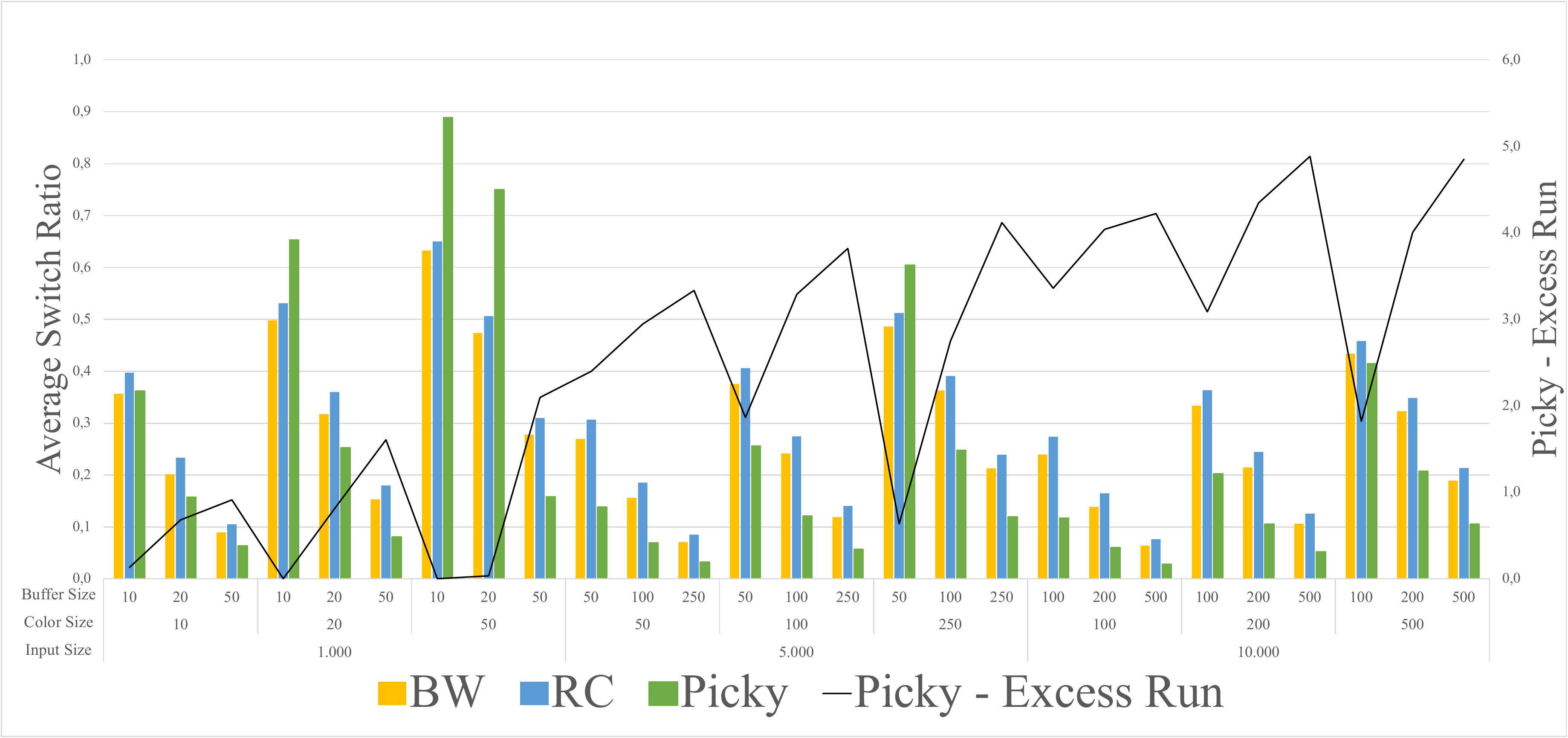}
\caption{Output/Input Switch Ratios for Zipf Distribution with $a=1,1$}
\end{figure}

\begin{figure}[htbp]
\centering
\includegraphics[width=1.0\textwidth]{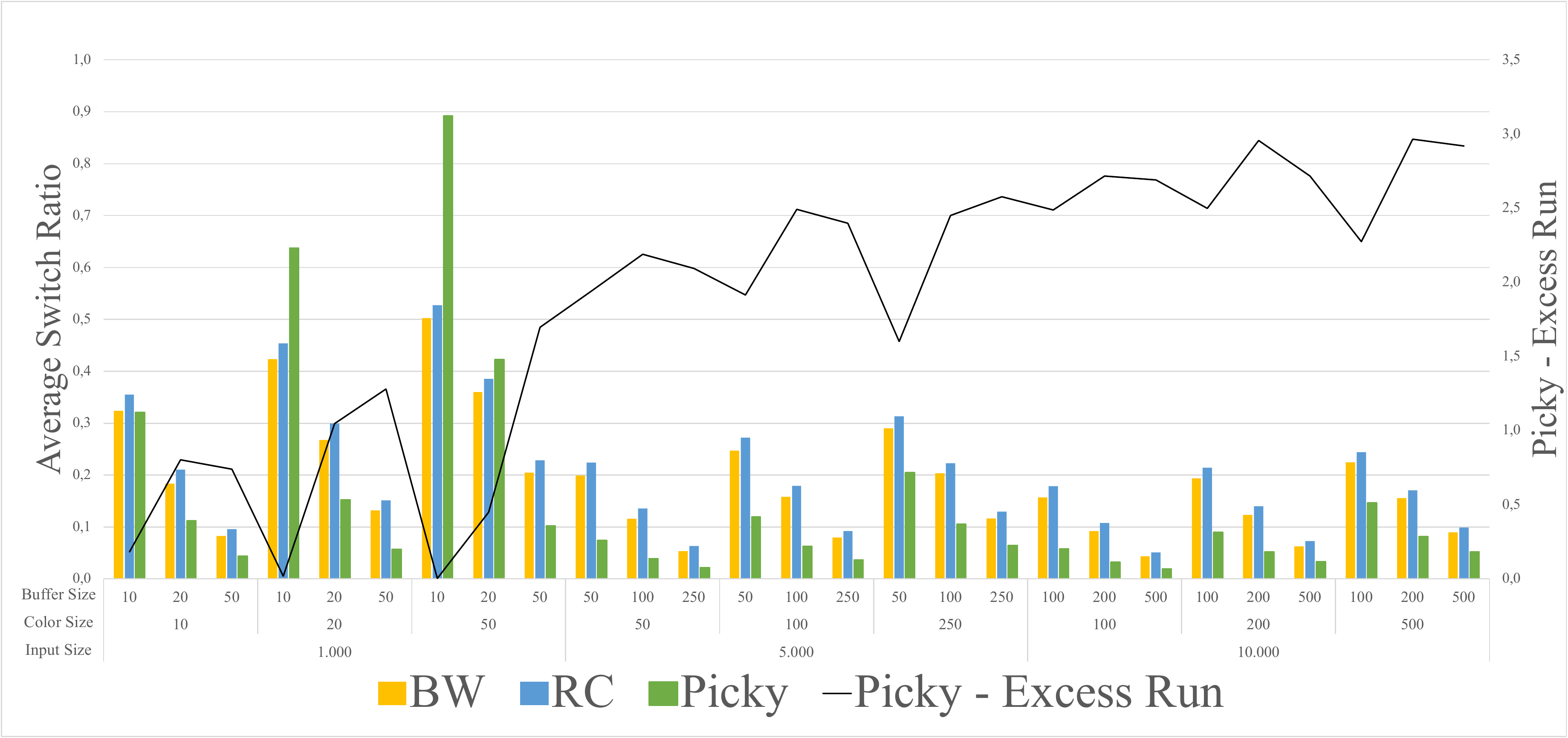}
\caption{Output/Input Switch Ratios for Zipf Distribution with $a=1,5$}
\end{figure}

\begin{figure}[htbp]
\centering
\includegraphics[width=1.0\textwidth]{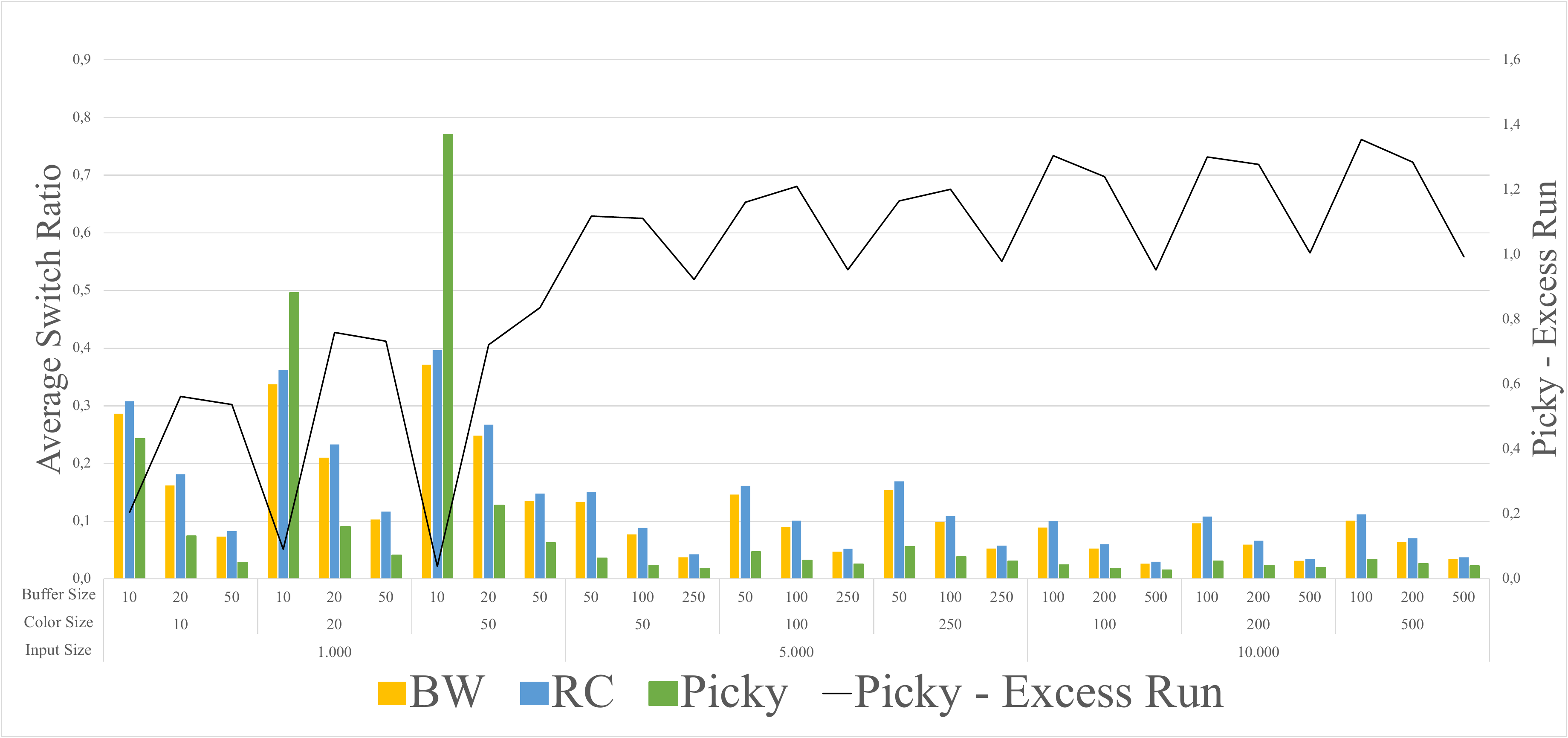}
\caption{Output/Input Switch Ratios for Zipf Distribution with $a=2$}
\end{figure}

\begin{table}[htbp]
    \scriptsize
    \caption{Average Results for Uniform Distribution}
    \centering
    \begin{tabular}{|c|c|c|c|c|c|c|}
    \hline
    \multicolumn{3}{|c|}{Test Parameters} & \multicolumn{3}{c}{Average Switch Ratio} & \multicolumn{1}{|c|}{Average Excess Run} \\
    \hline
    Input Size & Color Size & Buffer Size & BW & RC & Picky & Picky \\
    \hline
    \multirow{9}{*}{1.000}&\multirow{3}{*}{10}&10&0,400 & 0,447 & 0,376 & 0,07\\
    &&20&0,226 & 0,263 & 0,209 & 0,21\\
    &&50&0,100 & 0,118 & 0,093 & 0,30\\
    \cline{2-7}
    &\multirow{3}{*}{20}&10&0,617 & 0,644 & 0,576 & 0,00\\
    &&20&0,389 & 0,443 & 0,359 & 0,11\\
    &&50&0,187 & 0,223 & 0,158 & 0,75\\
    \cline{2-7}
    &\multirow{3}{*}{50}&10&0,828 & 0,835 & 0,821 & 0,00\\
    &&20&0,667 & 0,690 & 0,632 & 0,00\\
    &&50&0,388 & 0,442 & 0,352 & 0,18\\
    \hline
    
    \multirow{9}{*}{5.000}&\multirow{3}{*}{50}&50&0,381 & 0,435 & 0,343 & 0,20 \\
    &&100&0,220 & 0,262 & 0,194 & 0,61 \\
    &&250&0,099 & 0,121 & 0,086 & 0,92 \\
    \cline{2-7}
    &\multirow{3}{*}{100}&50&0,593 & 0,628 & 0,530 & 0,00 \\
    &&100&0,378 & 0,435 & 0,336 & 0,33 \\
    &&250&0,185 & 0,222 & 0,130 & 3,02 \\
    \cline{2-7}
    &\multirow{3}{*}{250}&50&0,814 & 0,823 & 0,805 & 0,00 \\
    &&100&0,657 & 0,685 & 0,611 & 0,00 \\
    &&250&0,385 & 0,439 & 0,343 & 0,29 \\
    \hline
    
    \multirow{9}{*}{10.000}&\multirow{3}{*}{100}&100&0,377 & 0,433 & 0,335 & 0,30 \\
    &&200&0,218 & 0,262 & 0,187 & 1,01 \\
    &&500&0,099 & 0,121 & 0,081 & 1,59 \\
    \cline{2-7}
    &\multirow{3}{*}{200}&100&0,590 & 0,627 & 0,523 & 0,00 \\
    &&200&0,379 & 0,434 & 0,336 & 0,28 \\
    &&500&0,184 & 0,222 & 0,122 & 4,22 \\
    \cline{2-7}
    &\multirow{3}{*}{500}&100&0,814 & 0,822 & 0,803 & 0,00 \\
    &&200&0,656 & 0,683 & 0,609 & 0,00 \\
    &&500&0,383 & 0,438 & 0,339 & 0,38 \\
    \hline
    \multicolumn{3}{|r|}{Average}&0,415 & 0,452 & 0,381 & 0,55\\
    \hline
    \multicolumn{3}{|r|}{Number of Best Cases}&0&0&27&\multicolumn{1}{r}{}\\
    \cline{1-6}
    \end{tabular}
\end{table}

\begin{table}[htbp]
    \scriptsize
    \caption{Average Results for Binomial Distribution with $p=0,3$}
    \centering
    \begin{tabular}{|c|c|c|c|c|c|c|}
    \hline
    \multicolumn{3}{|c|}{Test Parameters} & \multicolumn{3}{c}{Average Switch Ratio} & \multicolumn{1}{|c|}{Average Excess Run} \\
    \hline
    Input Size & Color Size & Buffer Size & BW & RC & Picky & Picky \\
    \hline
    \multirow{9}{*}{1.000}&\multirow{3}{*}{10}&10&0,281 & 0,314 & 0,245 & 0,32 \\
    &&20&0,156 & 0,178 & 0,113 & 0,71 \\
    &&50&0,068 & 0,078 & 0,040 & 1,04 \\
    \cline{2-7}
    &\multirow{3}{*}{20}&10&0,366 & 0,405 & 0,383 & 0,12 \\
    &&20&0,213 & 0,247 & 0,170 & 0,71 \\
    &&50&0,098 & 0,113 & 0,057 & 1,54 \\
    \cline{2-7}
    &\multirow{3}{*}{50}&10&0,490 & 0,530 & 0,521 & 0,01 \\
    &&20&0,306 & 0,345 & 0,258 & 0,72 \\
    &&50&0,147 & 0,172 & 0,095 & 1,88 \\
    \hline
    
    \multirow{9}{*}{5.000}&\multirow{3}{*}{50}&50&0,143 & 0,169 & 0,086 & 2,06 \\
    &&100&0,077 & 0,093 & 0,035 & 3,01 \\
    &&250&0,033 & 0,040 & 0,013 & 3,31 \\
    \cline{2-7}
    &\multirow{3}{*}{100}&50&0,191 & 0,225 & 0,125 & 2,06 \\
    &&100&0,105 & 0,127 & 0,053 & 3,34 \\
    &&250&0,046 & 0,056 & 0,021 & 3,88 \\
    \cline{2-7}
    &\multirow{3}{*}{250}&50&0,273 & 0,317 & 0,209 & 1,44 \\
    &&100&0,156 & 0,187 & 0,092 & 3,29 \\
    &&250&0,071 & 0,086 & 0,038 & 3,87 \\
    \hline
    
    \multirow{9}{*}{10.000}&\multirow{3}{*}{100}&100&0,105 & 0,126 & 0,053 & 3,34 \\
    &&200&0,056 & 0,068 & 0,024 & 4,09 \\
    &&500&0,024 & 0,029 & 0,011 & 4,00 \\
    \cline{2-7}
    &\multirow{3}{*}{200}&100&0,141 & 0,170 & 0,078 & 3,46 \\
    &&200&0,077 & 0,093 & 0,038 & 4,17 \\
    &&500&0,033 & 0,041 & 0,017 & 4,07 \\
    \cline{2-7}
    &\multirow{3}{*}{500}&100&0,207 & 0,245 & 0,130 & 3,28 \\
    &&200&0,116 & 0,140 & 0,063 & 4,39 \\
    &&500&0,052 & 0,063 & 0,028 & 4,58 \\
    \hline
    \multicolumn{3}{|r|}{Average}&0,149 & 0,172 & 0,111 & 2,54 \\
    \hline
    \multicolumn{3}{|r|}{Number of Best Cases}&2 &0 &25 &\multicolumn{1}{r}{}\\
    \cline{1-6}
    \end{tabular}
\end{table}

\begin{table}[htbp]
    \scriptsize
    \caption{Average Results for Binomial Distribution with $p=0,5$}
    \centering
    \begin{tabular}{|c|c|c|c|c|c|c|}
    \hline
    \multicolumn{3}{|c|}{Test Parameters} & \multicolumn{3}{c}{Average Switch Ratio} & \multicolumn{1}{|c|}{Average Excess Run} \\
    \hline
    Input Size & Color Size & Buffer Size & BW & RC & Picky & Picky \\
    \hline
    \multirow{9}{*}{1.000}&\multirow{3}{*}{10}&10&0,302 & 0,336 & 0,269 & 0,31 \\
    &&20&0,169 & 0,193 & 0,128 & 0,68 \\
    &&50&0,075 & 0,087 & 0,046 & 1,11 \\
    \cline{2-7}
    &\multirow{3}{*}{20}&10&0,387 & 0,429 & 0,401 & 0,08 \\
    &&20&0,227 & 0,261 & 0,187 & 0,64 \\
    &&50&0,105 & 0,123 & 0,063 & 1,62 \\
    \cline{2-7}
    &\multirow{3}{*}{50}&10&0,516 & 0,555 & 0,491 & 0,01 \\
    &&20&0,324 & 0,369 & 0,277 & 0,68 \\
    &&50&0,158 & 0,185 & 0,101 & 1,94 \\
    \hline
    
    \multirow{9}{*}{5.000}&\multirow{3}{*}{50}&50&0,153 & 0,182 & 0,096 & 2,03 \\
    &&100&0,083 & 0,100 & 0,040 & 3,10 \\
    &&250&0,036 & 0,044 & 0,014 & 3,56 \\
    \cline{2-7}
    &\multirow{3}{*}{100}&50&0,205 & 0,242 & 0,139 & 2,01 \\
    &&100&0,114 & 0,137 & 0,059 & 3,35 \\
    &&250&0,051 & 0,061 & 0,024 & 3,75 \\
    \cline{2-7}
    &\multirow{3}{*}{250}&50&0,292 & 0,336 & 0,221 & 1,59 \\
    &&100&0,169 & 0,201 & 0,105 & 3,10 \\
    &&250&0,077 & 0,093 & 0,042 & 3,97 \\
    \hline
    
    \multirow{9}{*}{10.000}&\multirow{3}{*}{100}&100&0,113 & 0,136 & 0,058 & 3,34 \\
    &&200&0,060 & 0,073 & 0,027 & 4,03 \\
    &&500&0,026 & 0,032 & 0,012 & 3,96 \\
    \cline{2-7}
    &\multirow{3}{*}{200}&100&0,152 & 0,182 & 0,089 & 3,28 \\
    &&200&0,083 & 0,101 & 0,042 & 4,10 \\
    &&500&0,036 & 0,044 & 0,018 & 4,18 \\
    \cline{2-7}
    &\multirow{3}{*}{500}&100&0,221 & 0,261 & 0,146 & 3,05 \\
    &&200&0,125 & 0,151 & 0,068 & 4,48 \\
    &&500&0,056 & 0,068 & 0,030 & 4,59 \\
    \hline
    \multicolumn{3}{|r|}{Average}&0,160 & 0,185 & 0,118 & 2,54 \\
    \hline
    \multicolumn{3}{|r|}{Number of Best Cases}&1 &0 &26 &\multicolumn{1}{r}{}\\
    \cline{1-6}
    \end{tabular}
\end{table}

\begin{table}[htbp]
    \caption{Average Results for Binomial Distribution with $p=0,7$}
    \scriptsize
    \centering
    \begin{tabular}{|c|c|c|c|c|c|c|}
    \hline
    \multicolumn{3}{|c|}{Test Parameters} & \multicolumn{3}{c}{Average Switch Ratio} & \multicolumn{1}{|c|}{Average Excess Run} \\
    \hline
    Input Size & Color Size & Buffer Size & BW & RC & Picky & Picky \\
    \hline
    \multirow{9}{*}{1.000}&\multirow{3}{*}{10}&10&0,280 & 0,310 & 0,240 & 0,31 \\
    &&20&0,154 & 0,175 & 0,112 & 0,70 \\
    &&50&0,067 & 0,078 & 0,041 & 1,04 \\
    \cline{2-7}
    &\multirow{3}{*}{20}&10&0,366 & 0,403 & 0,370 & 0,14 \\
    &&20&0,211 & 0,243 & 0,169 & 0,72 \\
    &&50&0,097 & 0,114 & 0,059 & 1,44 \\
    \cline{2-7}
    &\multirow{3}{*}{50}&10&0,492 & 0,532 & 0,634 & 0,00 \\
    &&20&0,307 & 0,349 & 0,270 & 0,64 \\
    &&50&0,147 & 0,173 & 0,091 & 1,97 \\
    \hline
    
    \multirow{9}{*}{5.000}&\multirow{3}{*}{50}&50&0,143 & 0,169 & 0,088 & 2,00 \\
    &&100&0,077 & 0,092 & 0,035 & 3,01 \\
    &&250&0,033 & 0,040 & 0,013 & 3,30 \\
    \cline{2-7}
    &\multirow{3}{*}{100}&50&0,191 & 0,225 & 0,126 & 2,02 \\
    &&100&0,105 & 0,126 & 0,054 & 3,24 \\
    &&250&0,046 & 0,056 & 0,022 & 3,66 \\
    \cline{2-7}
    &\multirow{3}{*}{250}&50&0,273 & 0,316 & 0,210 & 1,46 \\
    &&100&0,157 & 0,188 & 0,092 & 3,34 \\
    &&250&0,071 & 0,086 & 0,038 & 3,86 \\
    \hline
    
    \multirow{9}{*}{10.000}&\multirow{3}{*}{100}&100&0,105 & 0,126 & 0,053 & 3,27 \\
    &&200&0,056 & 0,068 & 0,024 & 3,98 \\
    &&500&0,024 & 0,029 & 0,010 & 3,99 \\
    \cline{2-7}
    &\multirow{3}{*}{200}&100&0,141 & 0,169 & 0,079 & 3,38 \\
    &&200&0,077 & 0,093 & 0,038 & 4,05 \\
    &&500&0,033 & 0,040 & 0,017 & 4,09 \\
    \cline{2-7}
    &\multirow{3}{*}{500}&100&0,207 & 0,245 & 0,130 & 3,26 \\
    &&200&0,116 & 0,140 & 0,063 & 4,39 \\
    &&500&0,051 & 0,063 & 0,028 & 4,23 \\
    \hline
    \multicolumn{3}{|r|}{Average}&0,149 & 0,172 & 0,115 & 2,50 \\
    \hline
    \multicolumn{3}{|r|}{Number of Best Cases}&2 &0 &25 &\multicolumn{1}{r}{}\\
    \cline{1-6}
    \end{tabular}
\end{table}

\begin{table}[htbp]
    \caption{Average Results for Negative Binomial Distribution with $p=0,3$}
    \scriptsize
    \centering
    \begin{tabular}{|c|c|c|c|c|c|c|}
    \hline
    \multicolumn{3}{|c|}{Test Parameters} & \multicolumn{3}{c}{Average Switch Ratio} & \multicolumn{1}{|c|}{Average Excess Run} \\
    \hline
    Input Size & Color Size & Buffer Size & BW & RC & Picky & Picky \\
    \hline
    \multirow{9}{*}{1.000}&\multirow{3}{*}{10}&10&0,284 & 0,317 & 0,243 & 0,31 \\
    &&20&0,158 & 0,181 & 0,105 & 0,82 \\
    &&50&0,070 & 0,081 & 0,032 & 1,33 \\
    \cline{2-7}
    &\multirow{3}{*}{20}&10&0,327 & 0,361 & 0,295 & 0,29 \\
    &&20&0,189 & 0,214 & 0,121 & 0,97 \\
    &&50&0,086 & 0,100 & 0,032 & 1,66 \\
    \cline{2-7}
    &\multirow{3}{*}{50}&10&0,357 & 0,393 & 0,336 & 0,25 \\
    &&20&0,212 & 0,242 & 0,134 & 1,11 \\
    &&50&0,099 & 0,116 & 0,036 & 1,91 \\
    \hline
    
    \multirow{9}{*}{5.000}&\multirow{3}{*}{50}&50&0,096 & 0,113 & 0,023 & 2,21 \\
    &&100&0,051 & 0,061 & 0,009 & 2,21 \\
    &&250&0,022 & 0,026 & 0,006 & 2,02 \\
    \cline{2-7}
    &\multirow{3}{*}{100}&50&0,102 & 0,120 & 0,024 & 2,34 \\
    &&100&0,055 & 0,065 & 0,010 & 2,34 \\
    &&250&0,024 & 0,028 & 0,007 & 2,08 \\
    \cline{2-7}
    &\multirow{3}{*}{250}&50&0,106 & 0,124 & 0,025 & 2,41 \\
    &&100&0,057 & 0,068 & 0,010 & 2,38 \\
    &&250&0,025 & 0,030 & 0,007 & 2,23 \\
    \hline
    
    \multirow{9}{*}{10.000}&\multirow{3}{*}{100}&100&0,054 & 0,065 & 0,007 & 2,43 \\
    &&200&0,029 & 0,034 & 0,004 & 2,37 \\
    &&500&0,012 & 0,015 & 0,004 & 2,21 \\
    \cline{2-7}
    &\multirow{3}{*}{200}&100&0,056 & 0,067 & 0,008 & 2,49 \\
    &&200&0,030 & 0,036 & 0,005 & 2,38 \\
    &&500&0,013 & 0,015 & 0,004 & 2,15 \\
    \cline{2-7}
    &\multirow{3}{*}{500}&100&0,058 & 0,069 & 0,008 & 2,55 \\
    &&200&0,031 & 0,036 & 0,005 & 2,45 \\
    &&500&0,013 & 0,016 & 0,004 & 2,15 \\
    \hline
    \multicolumn{3}{|r|}{Average}&0,097 & 0,111 & 0,056 & 1,85 \\
    \hline
    \multicolumn{3}{|r|}{Number of Best Cases}&0 &0 &27 &\multicolumn{1}{r}{}\\
    \cline{1-6}
    \end{tabular}
\end{table}

\begin{table}[htbp]
    \caption{Average Results for Negative Binomial Distribution with $p=0,5$}
    \scriptsize
    \centering
    \begin{tabular}{|c|c|c|c|c|c|c|}
    \hline
    \multicolumn{3}{|c|}{Test Parameters} & \multicolumn{3}{c}{Average Switch Ratio} & \multicolumn{1}{|c|}{Average Excess Run} \\
    \hline
    Input Size & Color Size & Buffer Size & BW & RC & Picky & Picky \\
    \hline
    \multirow{9}{*}{1.000}&\multirow{3}{*}{10}&10&0,347 & 0,387 & 0,316 & 0,22 \\
    &&20&0,196 & 0,228 & 0,165 & 0,53 \\
    &&50&0,087 & 0,102 & 0,069 & 0,71 \\
    \cline{2-7}
    &\multirow{3}{*}{20}&10&0,460 & 0,498 & 0,437 & 0,03 \\
    &&20&0,275 & 0,318 & 0,231 & 0,63 \\
    &&50&0,129 & 0,152 & 0,092 & 1,36 \\
    \cline{2-7}
    &\multirow{3}{*}{50}&10&0,600 & 0,631 & 0,682 & 0,00 \\
    &&20&0,396 & 0,441 & 0,402 & 0,11 \\
    &&50&0,199 & 0,235 & 0,142 & 1,68 \\
    \hline
    
    \multirow{9}{*}{5.000}&\multirow{3}{*}{50}&50&0,194 & 0,228 & 0,132 & 1,88 \\
    &&100&0,106 & 0,128 & 0,059 & 3,02 \\
    &&250&0,047 & 0,056 & 0,023 & 3,54 \\
    \cline{2-7}
    &\multirow{3}{*}{100}&50&0,260 & 0,303 & 0,203 & 1,24 \\
    &&100&0,147 & 0,176 & 0,087 & 3,22 \\
    &&250&0,066 & 0,080 & 0,035 & 3,81 \\
    \cline{2-7}
    &\multirow{3}{*}{250}&50&0,366 & 0,414 & 0,325 & 0,72 \\
    &&100&0,218 & 0,257 & 0,145 & 3,05 \\
    &&250&0,102 & 0,123 & 0,057 & 4,43 \\
    \hline
    
    \multirow{9}{*}{10.000}&\multirow{3}{*}{100}&100&0,146 & 0,175 & 0,085 & 3,29 \\
    &&200&0,079 & 0,096 & 0,041 & 3,94 \\
    &&500&0,034 & 0,042 & 0,018 & 3,94 \\
    \cline{2-7}
    &\multirow{3}{*}{200}&100&0,198 & 0,235 & 0,124 & 3,21 \\
    &&200&0,109 & 0,133 & 0,061 & 4,15 \\
    &&500&0,048 & 0,059 & 0,026 & 4,40 \\
    \cline{2-7}
    &\multirow{3}{*}{500}&100&0,286 & 0,331 & 0,208 & 2,32 \\
    &&200&0,165 & 0,198 & 0,098 & 4,40 \\
    &&500&0,075 & 0,092 & 0,043 & 4,65 \\
    \hline
    \multicolumn{3}{|r|}{Average}&0,198 & 0,227 & 0,160 & 2,39 \\
    \hline
    \multicolumn{3}{|r|}{Number of Best Cases}&2 &0 &25 &\multicolumn{1}{r}{}\\
    \cline{1-6}
    \end{tabular}
\end{table}

\begin{table}[htbp]
    \caption{Average Results for Negative Binomial Distribution with $p=0,7$}
    \scriptsize
    \centering
    \begin{tabular}{|c|c|c|c|c|c|c|}
    \hline
    \multicolumn{3}{|c|}{Test Parameters} & \multicolumn{3}{c}{Average Switch Ratio} & \multicolumn{1}{|c|}{Average Excess Run} \\
    \hline
    Input Size & Color Size & Buffer Size & BW & RC & Picky & Picky \\
    \hline
    \multirow{9}{*}{1.000}&\multirow{3}{*}{10}&10&0,369 & 0,412 & 0,365 & 0,12 \\
    &&20&0,208 & 0,242 & 0,181 & 0,50 \\
    &&50&0,093 & 0,109 & 0,072 & 0,81 \\
    \cline{2-7}
    &\multirow{3}{*}{20}&10&0,497 & 0,538 & 0,533 & 0,01 \\
    &&20&0,308 & 0,350 & 0,263 & 0,70 \\
    &&50&0,149 & 0,174 & 0,093 & 1,87 \\
    \cline{2-7}
    &\multirow{3}{*}{50}&10&0,635 & 0,660 & 0,601 & 0,00 \\
    &&20&0,433 & 0,475 & 0,452 & 0,04 \\
    &&50&0,226 & 0,263 & 0,169 & 1,63 \\
    \hline
    
    \multirow{9}{*}{5.000}&\multirow{3}{*}{50}&50&0,220 & 0,257 & 0,160 & 1,75 \\
    &&100&0,123 & 0,148 & 0,067 & 3,30 \\
    &&250&0,055 & 0,066 & 0,027 & 3,89 \\
    \cline{2-7}
    &\multirow{3}{*}{100}&50&0,287 & 0,333 & 0,219 & 1,55 \\
    &&100&0,166 & 0,198 & 0,103 & 3,03 \\
    &&250&0,076 & 0,091 & 0,041 & 3,85 \\
    \cline{2-7}
    &\multirow{3}{*}{250}&50&0,397 & 0,443 & 0,421 & 0,13 \\
    &&100&0,241 & 0,283 & 0,176 & 2,35 \\
    &&250&0,115 & 0,138 & 0,065 & 4,40 \\
    \hline
    
    \multirow{9}{*}{10.000}&\multirow{3}{*}{100}&100&0,165 & 0,197 & 0,100 & 3,14 \\
    &&200&0,090 & 0,110 & 0,047 & 4,11 \\
    &&500&0,040 & 0,048 & 0,021 & 4,20 \\
    \cline{2-7}
    &\multirow{3}{*}{200}&100&0,220 & 0,259 & 0,145 & 3,05 \\
    &&200&0,123 & 0,149 & 0,068 & 4,39 \\
    &&500&0,055 & 0,067 & 0,030 & 4,54 \\
    \cline{2-7}
    &\multirow{3}{*}{500}&100&0,311 & 0,358 & 0,225 & 2,61 \\
    &&200&0,182 & 0,217 & 0,108 & 4,51 \\
    &&500&0,084 & 0,102 & 0,047 & 5,04 \\
    \hline
    \multicolumn{3}{|r|}{Average}&0,217 & 0,248 & 0,178 & 2,43 \\
    \hline
    \multicolumn{3}{|r|}{Number of Best Cases}&3 &0 &24 &\multicolumn{1}{r}{}\\
    \cline{1-6}
    \end{tabular}
\end{table}

\begin{table}[htbp]
    \caption{Average Results for Geometric Distribution with $p=0,3$}
    \scriptsize
    \centering
    \begin{tabular}{|c|c|c|c|c|c|c|}
    \hline
    \multicolumn{3}{|c|}{Test Parameters} & \multicolumn{3}{c}{Average Switch Ratio} & \multicolumn{1}{|c|}{Average Excess Run} \\
    \hline
    Input Size & Color Size & Buffer Size & BW & RC & Picky & Picky \\
    \hline
    \multirow{9}{*}{1.000}&\multirow{3}{*}{10}&10&0,334 & 0,373 & 0,339 & 0,20 \\
    &&20&0,191 & 0,220 & 0,138 & 0,80 \\
    &&50&0,085 & 0,099 & 0,044 & 1,31 \\
    \cline{2-7}
    &\multirow{3}{*}{20}&10&0,373 & 0,408 & 0,513 & 0,06 \\
    &&20&0,224 & 0,254 & 0,147 & 1,06 \\
    &&50&0,106 & 0,123 & 0,037 & 1,97 \\
    \cline{2-7}
    &\multirow{3}{*}{50}&10&0,371 & 0,404 & 0,531 & 0,06 \\
    &&20&0,225 & 0,254 & 0,143 & 1,11 \\
    &&50&0,108 & 0,124 & 0,039 & 1,86 \\
    \hline
    
    \multirow{9}{*}{5.000}&\multirow{3}{*}{50}&50&0,104 & 0,122 & 0,023 & 2,28 \\
    &&100&0,057 & 0,067 & 0,010 & 2,27 \\
    &&250&0,025 & 0,029 & 0,007 & 2,06 \\
    \cline{2-7}
    &\multirow{3}{*}{100}&50&0,105 & 0,122 & 0,023 & 2,31 \\
    &&100&0,057 & 0,066 & 0,010 & 2,27 \\
    &&250&0,025 & 0,029 & 0,007 & 2,02 \\
    \cline{2-7}
    &\multirow{3}{*}{250}&50&0,104 & 0,122 & 0,023 & 2,28 \\
    &&100&0,056 & 0,066 & 0,010 & 2,24 \\
    &&250&0,025 & 0,029 & 0,007 & 2,09 \\
    \hline
    
    \multirow{9}{*}{10.000}&\multirow{3}{*}{100}&100&0,056 & 0,066 & 0,007 & 2,36 \\
    &&200&0,029 & 0,035 & 0,005 & 2,32 \\
    &&500&0,013 & 0,015 & 0,004 & 2,14 \\
    \cline{2-7}
    &\multirow{3}{*}{200}&100&0,056 & 0,066 & 0,008 & 2,35 \\
    &&200&0,030 & 0,035 & 0,005 & 2,31 \\
    &&500&0,013 & 0,015 & 0,004 & 2,11 \\
    \cline{2-7}
    &\multirow{3}{*}{500}&100&0,056 & 0,066 & 0,007 & 2,33 \\
    &&200&0,030 & 0,035 & 0,005 & 2,38 \\
    &&500&0,013 & 0,015 & 0,004 & 2,09 \\
    \hline
    \multicolumn{3}{|r|}{Average}&0,106 & 0,121 & 0,078 & 1,80 \\
    \hline
    \multicolumn{3}{|r|}{Number of Best Cases}&3 &0 &24 &\multicolumn{1}{r}{}\\
    \cline{1-6}
    \end{tabular}
\end{table}

\begin{table}[htbp]
    \caption{Average Results for Geometric Distribution with $p=0,5$}
    \scriptsize
    \centering
    \begin{tabular}{|c|c|c|c|c|c|c|}
    \hline
    \multicolumn{3}{|c|}{Test Parameters} & \multicolumn{3}{c}{Average Switch Ratio} & \multicolumn{1}{|c|}{Average Excess Run} \\
    \hline
    Input Size & Color Size & Buffer Size & BW & RC & Picky & Picky \\
    \hline
    \multirow{9}{*}{1.000}&\multirow{3}{*}{10}&10&0,250 & 0,276 & 0,220 & 0,23 \\
    &&20&0,141 & 0,158 & 0,065 & 0,79 \\
    &&50&0,064 & 0,071 & 0,018 & 0,92 \\
    \cline{2-7}
    &\multirow{3}{*}{20}&10&0,253 & 0,278 & 0,219 & 0,24 \\
    &&20&0,143 & 0,161 & 0,064 & 0,81 \\
    &&50&0,065 & 0,073 & 0,017 & 0,94 \\
    \cline{2-7}
    &\multirow{3}{*}{50}&10&0,250 & 0,274 & 0,218 & 0,24 \\
    &&20&0,141 & 0,159 & 0,064 & 0,80 \\
    &&50&0,065 & 0,073 & 0,018 & 0,95 \\
    \hline
    
    \multirow{9}{*}{5.000}&\multirow{3}{*}{50}&50&0,062 & 0,071 & 0,007 & 1,16 \\
    &&100&0,033 & 0,038 & 0,004 & 0,95 \\
    &&250&0,014 & 0,016 & 0,004 & 0,85 \\
    \cline{2-7}
    &\multirow{3}{*}{100}&50&0,063 & 0,071 & 0,007 & 1,14 \\
    &&100&0,033 & 0,038 & 0,004 & 0,96 \\
    &&250&0,014 & 0,016 & 0,004 & 0,84 \\
    \cline{2-7}
    &\multirow{3}{*}{250}&50&0,063 & 0,072 & 0,007 & 1,15 \\
    &&100&0,033 & 0,038 & 0,004 & 0,98 \\
    &&250&0,014 & 0,016 & 0,004 & 0,85 \\
    \hline
    
    \multirow{9}{*}{10.000}&\multirow{3}{*}{100}&100&0,032 & 0,037 & 0,002 & 1,03 \\
    &&200&0,017 & 0,019 & 0,002 & 0,95 \\
    &&500&0,007 & 0,008 & 0,002 & 0,85 \\
    \cline{2-7}
    &\multirow{3}{*}{200}&100&0,032 & 0,037 & 0,002 & 0,99 \\
    &&200&0,017 & 0,019 & 0,002 & 0,94 \\
    &&500&0,007 & 0,008 & 0,002 & 0,85 \\
    \cline{2-7}
    &\multirow{3}{*}{500}&100&0,033 & 0,037 & 0,002 & 1,01 \\
    &&200&0,017 & 0,020 & 0,002 & 0,95 \\
    &&500&0,007 & 0,008 & 0,002 & 0,86 \\
    \hline
    \multicolumn{3}{|r|}{Average}&0,069 & 0,078 & 0,036 & 0,86 \\
    \hline
    \multicolumn{3}{|r|}{Number of Best Cases}&0 &0 &27 &\multicolumn{1}{r}{}\\
    \cline{1-6}
    \end{tabular}
\end{table}

\begin{table}[htbp]
    \caption{Average Results for Geometric Distribution with $p=0,7$}
    \scriptsize
    \centering
    \begin{tabular}{|c|c|c|c|c|c|c|}
    \hline
    \multicolumn{3}{|c|}{Test Parameters} & \multicolumn{3}{c}{Average Switch Ratio} & \multicolumn{1}{|c|}{Average Excess Run} \\
    \hline
    Input Size & Color Size & Buffer Size & BW & RC & Picky & Picky \\
    \hline
    \multirow{9}{*}{1.000}&\multirow{3}{*}{10}&10&0,182 & 0,196 & 0,177 & 0,06 \\
    &&20&0,099 & 0,108 & 0,060 & 0,25 \\
    &&50&0,044 & 0,048 & 0,017 & 0,33 \\
    \cline{2-7}
    &\multirow{3}{*}{20}&10&0,182 & 0,194 & 0,181 & 0,06 \\
    &&20&0,099 & 0,107 & 0,055 & 0,27 \\
    &&50&0,044 & 0,048 & 0,015 & 0,35 \\
    \cline{2-7}
    &\multirow{3}{*}{50}&10&0,184 & 0,197 & 0,179 & 0,06 \\
    &&20&0,100 & 0,109 & 0,056 & 0,27 \\
    &&50&0,044 & 0,048 & 0,015 & 0,35 \\
    \hline
    
    \multirow{9}{*}{5.000}&\multirow{3}{*}{50}&50&0,042 & 0,046 & 0,009 & 0,42 \\
    &&100&0,022 & 0,024 & 0,003 & 0,42 \\
    &&250&0,010 & 0,011 & 0,003 & 0,32 \\
    \cline{2-7}
    &\multirow{3}{*}{100}&50&0,042 & 0,046 & 0,009 & 0,43 \\
    &&100&0,022 & 0,024 & 0,004 & 0,42 \\
    &&250&0,010 & 0,011 & 0,003 & 0,33 \\
    \cline{2-7}
    &\multirow{3}{*}{250}&50&0,042 & 0,046 & 0,008 & 0,42 \\
    &&100&0,022 & 0,024 & 0,003 & 0,41 \\
    &&250&0,010 & 0,011 & 0,003 & 0,33 \\
    \hline
    
    \multirow{9}{*}{10.000}&\multirow{3}{*}{100}&100&0,022 & 0,024 & 0,002 & 0,44 \\
    &&200&0,011 & 0,012 & 0,002 & 0,39 \\
    &&500&0,005 & 0,005 & 0,002 & 0,33 \\
    \cline{2-7}
    &\multirow{3}{*}{200}&100&0,021 & 0,024 & 0,002 & 0,43 \\
    &&200&0,011 & 0,012 & 0,002 & 0,39 \\
    &&500&0,005 & 0,005 & 0,002 & 0,32 \\
    \cline{2-7}
    &\multirow{3}{*}{500}&100&0,022 & 0,024 & 0,002 & 0,42 \\
    &&200&0,011 & 0,012 & 0,002 & 0,40 \\
    &&500&0,005 & 0,005 & 0,002 & 0,33 \\
    \hline
    \multicolumn{3}{|r|}{Average}&0,049 & 0,053 & 0,030 & 0,33 \\
    \hline
    \multicolumn{3}{|r|}{Number of Best Cases}&0 & 0 & 27 &\multicolumn{1}{r}{}\\
    \cline{1-6}
    \end{tabular}
\end{table}

\begin{table}[htbp]
    \caption{Average Results for Poisson Distribution with $m=1$}
    \scriptsize
    \centering
    \begin{tabular}{|c|c|c|c|c|c|c|}
    \hline
    \multicolumn{3}{|c|}{Test Parameters} & \multicolumn{3}{c}{Average Switch Ratio} & \multicolumn{1}{|c|}{Average Excess Run} \\
    \hline
    Input Size & Color Size & Buffer Size & BW & RC & Picky & Picky \\
    \hline
    \multirow{9}{*}{1.000}&\multirow{3}{*}{10}&10&0,207 & 0,225 & 0,198 & 0,11 \\
    &&20&0,112 & 0,124 & 0,074 & 0,47 \\
    &&50&0,048 & 0,054 & 0,024 & 0,71 \\
    \cline{2-7}
    &\multirow{3}{*}{20}&10&0,206 & 0,226 & 0,199 & 0,10 \\
    &&20&0,111 & 0,123 & 0,075 & 0,47 \\
    &&50&0,048 & 0,054 & 0,024 & 0,70 \\
    \cline{2-7}
    &\multirow{3}{*}{50}&10&0,205 & 0,224 & 0,201 & 0,10 \\
    &&20&0,11 & 0,124 & 0,076 & 0,45 \\
    &&50&0,048 & 0,054 & 0,023 & 0,73 \\
    \hline
    
    \multirow{9}{*}{5.000}&\multirow{3}{*}{50}&50&0,046 & 0,052 & 0,021 & 0,79 \\
    &&100&0,024 & 0,027 & 0,005 & 1,01 \\
    &&250&0,010 & 0,011 & 0,002 & 0,94 \\
    \cline{2-7}
    &\multirow{3}{*}{100}&50&0,046 & 0,052 & 0,021 & 0,77 \\
    &&100&0,024 & 0,027 & 0,005 & 1,01 \\
    &&250&0,010 & 0,012 & 0,002 & 0,93 \\
    \cline{2-7}
    &\multirow{3}{*}{250}&50&0,046 & 0,052 & 0,021 & 0,78\\
    &&100&0,024 & 0,027 & 0,005 & 1,00 \\
    &&250&0,010 & 0,011 & 0,002 & 0,95 \\
    \hline
    
    \multirow{9}{*}{10.000}&\multirow{3}{*}{100}&100&0,024 & 0,027 & 0,004 & 1,05 \\
    &&200&0,012 & 0,014 & 0,001 & 0,99 \\
    &&500&0,005 & 0,006 & 0,001 & 0,93 \\
    \cline{2-7}
    &\multirow{3}{*}{200}&100&0,024 & 0,027 & 0,004 & 1,04 \\
    &&200&0,012 & 0,014 & 0,001 & 1,01 \\
    &&500&0,005 & 0,006 & 0,001 & 0,95 \\
    \cline{2-7}
    &\multirow{3}{*}{500}&100&0,024 & 0,027 & 0,004 & 1,06 \\
    &&200&0,012 & 0,014 & 0,001 & 1,01 \\
    &&500&0,005 & 0,006 & 0,001 & 0,95 \\
    \hline
    \multicolumn{3}{|r|}{Average}&0,054 & 0,06 & 0,037 & 0,78 \\
    \hline
    \multicolumn{3}{|r|}{Number of Best Cases}&0 &0 &27 &\multicolumn{1}{r}{}\\
    \cline{1-6}
    \end{tabular}
\end{table}

\begin{table}[htbp]
    \caption{Average Results for Poisson Distribution with $m=2$}
    \scriptsize
    \centering
    \begin{tabular}{|c|c|c|c|c|c|c|}
    \hline
    \multicolumn{3}{|c|}{Test Parameters} & \multicolumn{3}{c}{Average Switch Ratio} & \multicolumn{1}{|c|}{Average Excess Run} \\
    \hline
    Input Size & Color Size & Buffer Size & BW & RC & Picky & Picky \\
    \hline
    \multirow{9}{*}{1.000}&\multirow{3}{*}{10}&10&0,275 & 0,307 & 0,242 & 0,29 \\
    &&20&0,152 & 0,172 & 0,112 & 0,68 \\
    &&50&0,067 & 0,077 & 0,039 & 1,05 \\
    \cline{2-7}
    &\multirow{3}{*}{20}&10&0,277 & 0,309 & 0,242 & 0,29 \\
    &&20&0,153 & 0,174 & 0,110 & 0,71 \\
    &&50&0,068 & 0,078 & 0,039 & 1,07 \\
    \cline{2-7}
    &\multirow{3}{*}{50}&10&0,276 & 0,306 & 0,247 & 0,26 \\
    &&20&0,152 & 0,174 & 0,112 & 0,69 \\
    &&50&0,068 & 0,078 & 0,038 & 1,09 \\
    \hline
    
    \multirow{9}{*}{5.000}&\multirow{3}{*}{50}&50&0,065 & 0,075 & 0,035 & 1,17 \\
    &&100&0,034 & 0,039 & 0,011 & 1,61 \\
    &&250&0,014 & 0,017 & 0,003 & 1,60 \\
    \cline{2-7}
    &\multirow{3}{*}{100}&50&0,065 & 0,075 & 0,035 & 1,18 \\
    &&100&0,034 & 0,039 & 0,012 & 1,59 \\
    &&250&0,014 & 0,017 & 0,004 & 1,61 \\
    \cline{2-7}
    &\multirow{3}{*}{250}&50&0,065 & 0,076 & 0,035 & 1,18 \\
    &&100&0,034 & 0,039 & 0,011 & 1,61 \\
    &&250&0,014 & 0,017 & 0,004 & 1,60 \\
    \hline
    
    \multirow{9}{*}{10.000}&\multirow{3}{*}{100}&100&0,033 & 0,039 & 0,011 & 1,64 \\
    &&200&0,017 & 0,020 & 0,003 & 1,79 \\
    &&500&0,007 & 0,008 & 0,002 & 1,59 \\
    \cline{2-7}
    &\multirow{3}{*}{200}&100&0,033 & 0,039 & 0,011 & 1,61 \\
    &&200&0,017 & 0,020 & 0,003 & 1,76 \\
    &&500&0,007 & 0,008 & 0,001 & 1,60 \\
    \cline{2-7}
    &\multirow{3}{*}{500}&100&0,033 & 0,039 & 0,011 & 1,65 \\
    &&200&0,017 & 0,020 & 0,003 & 1,80 \\
    &&500&0,007 & 0,008 & 0,002 & 1,58 \\
    \hline
    \multicolumn{3}{|r|}{Average}&0,074 & 0,084 & 0,051 & 1,27 \\
    \hline
    \multicolumn{3}{|r|}{Number of Best Cases}&0 &0 &27 &\multicolumn{1}{r}{}\\
    \cline{1-6}
    \end{tabular}
\end{table}

\begin{table}[htbp]
    \caption{Average Results for Poisson Distribution with $m=3$}
    \scriptsize
    \centering
    \begin{tabular}{|c|c|c|c|c|c|c|}
    \hline
    \multicolumn{3}{|c|}{Test Parameters} & \multicolumn{3}{c}{Average Switch Ratio} & \multicolumn{1}{|c|}{Average Excess Run} \\
    \hline
    Input Size & Color Size & Buffer Size & BW & RC & Picky & Picky \\
    \hline
    \multirow{9}{*}{1.000}&\multirow{3}{*}{10}&10&0,323 & 0,358 & 0,307 & 0,26 \\
    &&20&0,182 & 0,210 & 0,141 & 0,68 \\
    &&50&0,081 & 0,095 & 0,054 & 1,06 \\
    \cline{2-7}
    &\multirow{3}{*}{20}&10&0,325 & 0,363 & 0,319 & 0,23 \\
    &&20&0,184 & 0,211 & 0,142 & 0,70 \\
    &&50&0,083 & 0,097 & 0,052 & 1,15 \\
    \cline{2-7}
    &\multirow{3}{*}{50}&10&0,324 & 0,361 & 0,318 & 0,23 \\
    &&20&0,183 & 0,211 & 0,141 & 0,69 \\
    &&50&0,082 & 0,097 & 0,052 & 1,14 \\
    \hline
    
    \multirow{9}{*}{5.000}&\multirow{3}{*}{50}&50&0,080 & 0,095 & 0,048 & 1,28 \\
    &&100&0,042 & 0,049 & 0,016 & 1,96 \\
    &&250&0,018 & 0,021 & 0,005 & 2,06 \\
    \cline{2-7}
    &\multirow{3}{*}{100}&50&0,080 & 0,094 & 0,048 & 1,29 \\
    &&100&0,042 & 0,049 & 0,015 & 2,01 \\
    &&250&0,018 & 0,021 & 0,005 & 2,06 \\
    \cline{2-7}
    &\multirow{3}{*}{250}&50&0,080 & 0,094 & 0,047 & 1,30 \\
    &&100&0,042 & 0,049 & 0,015 & 1,97 \\
    &&250&0,018 & 0,021 & 0,005 & 2,06 \\
    \hline
    
    \multirow{9}{*}{10.000}&\multirow{3}{*}{100}&100&0,041 & 0,049 & 0,015 & 2,03 \\
    &&200&0,021 & 0,025 & 0,004 & 2,27 \\
    &&500&0,009 & 0,011 & 0,002 & 2,07 \\
    \cline{2-7}
    &\multirow{3}{*}{200}&100&0,041 & 0,049 & 0,014 & 2,04 \\
    &&200&0,021 & 0,025 & 0,004 & 2,28 \\
    &&500&0,009 & 0,011 & 0,002 & 2,00 \\
    \cline{2-7}
    &\multirow{3}{*}{500}&100&0,041 & 0,049 & 0,014 & 2,05 \\
    &&200&0,021 & 0,025 & 0,004 & 2,33 \\
    &&500&0,009 & 0,011 & 0,002 & 2,07 \\
    \hline
    \multicolumn{3}{|r|}{Average}&0,089 & 0,102 & 0,066 & 1,53 \\
    \hline
    \multicolumn{3}{|r|}{Number of Best Cases}&0 & 0 & 27 &\multicolumn{1}{r}{}\\
    \cline{1-6}
    \end{tabular}
\end{table}

\begin{table}[htbp]
    \caption{Average Results for Zipf Distribution with $a=1,1$}
    \scriptsize
    \centering
    \begin{tabular}{|c|c|c|c|c|c|c|}
    \hline
    \multicolumn{3}{|c|}{Test Parameters} & \multicolumn{3}{c}{Average Switch Ratio} & \multicolumn{1}{|c|}{Average Excess Run} \\
    \hline
    Input Size & Color Size & Buffer Size & BW & RC & Picky & Picky \\
    \hline
    \multirow{9}{*}{1.000}&\multirow{3}{*}{10}&10&0,357 & 0,397 & 0,362 & 0,13 \\
    &&20&0,202 & 0,234 & 0,158 & 0,68 \\
    &&50&0,090 & 0,105 & 0,064 & 0,91 \\
    \cline{2-7}
    &\multirow{3}{*}{20}&10&0,498 & 0,531 & 0,653 & 0,00 \\
    &&20&0,318 & 0,360 & 0,253 & 0,81 \\
    &&50&0,153 & 0,180 & 0,081 & 1,61 \\
    \cline{2-7}
    &\multirow{3}{*}{50}&10&0,633 & 0,65 & 0,889 & 0,00 \\
    &&20&0,474 & 0,506 & 0,750 & 0,04 \\
    &&50&0,278 & 0,310 & 0,158 & 2,10 \\
    \hline
    
    \multirow{9}{*}{5.000}&\multirow{3}{*}{50}&50&0,269 & 0,307 & 0,138 & 2,40 \\
    &&100&0,156 & 0,185 & 0,070 & 2,94 \\
    &&250&0,071 & 0,085 & 0,032 & 3,34 \\
    \cline{2-7}
    &\multirow{3}{*}{100}&50&0,375 & 0,406 & 0,256 & 1,87 \\
    &&100&0,241 & 0,274 & 0,122 & 3,29 \\
    &&250&0,119 & 0,141 & 0,057 & 3,82 \\
    \cline{2-7}
    &\multirow{3}{*}{250}&50&0,486 & 0,512 & 0,605 & 0,64 \\
    &&100&0,363 & 0,391 & 0,248 & 2,75 \\
    &&250&0,213 & 0,239 & 0,120 & 4,12 \\
    \hline
    
    \multirow{9}{*}{10.000}&\multirow{3}{*}{100}&100&0,240 & 0,274 & 0,117 & 3,36 \\
    &&200&0,139 & 0,165 & 0,061 & 4,04 \\
    &&500&0,064 & 0,077 & 0,029 & 4,22 \\
    \cline{2-7}
    &\multirow{3}{*}{200}&100&0,333 & 0,364 & 0,202 & 3,09 \\
    &&200&0,215 & 0,244 & 0,106 & 4,35 \\
    &&500&0,106 & 0,126 & 0,052 & 4,89 \\
    \cline{2-7}
    &\multirow{3}{*}{500}&100&0,434 & 0,458 & 0,415 & 1,82 \\
    &&200&0,323 & 0,349 & 0,207 & 4,01 \\
    &&500&0,189 & 0,213 & 0,106 & 4,85 \\
    \hline
    \multicolumn{3}{|r|}{Average}&0,272 & 0,299 & 0,234 & 2,45 \\
    \hline
    \multicolumn{3}{|r|}{Number of Best Cases}&5 & 0 & 22 &\multicolumn{1}{r}{}\\
    \cline{1-6}
    \end{tabular}
\end{table}

\begin{table}[htbp]
    \caption{Average Results for Zipf Distribution with $a=1,5$}
    \scriptsize
    \centering
    \begin{tabular}{|c|c|c|c|c|c|c|}
    \hline
    \multicolumn{3}{|c|}{Test Parameters} & \multicolumn{3}{c}{Average Switch Ratio} & \multicolumn{1}{|c|}{Average Excess Run} \\
    \hline
    Input Size & Color Size & Buffer Size & BW & RC & Picky & Picky \\
    \hline
    \multirow{9}{*}{1.000}&\multirow{3}{*}{10}&10&0,323 & 0,355 & 0,321 & 0,18 \\
    &&20&0,184 & 0,211 & 0,112 & 0,80 \\
    &&50&0,082 & 0,096 & 0,044 & 0,74 \\
    \cline{2-7}
    &\multirow{3}{*}{20}&10&0,423 & 0,454 & 0,638 & 0,02 \\
    &&20&0,268 & 0,300 & 0,153 & 1,05 \\
    &&50&0,132 & 0,151 & 0,057 & 1,28 \\
    \cline{2-7}
    &\multirow{3}{*}{50}&10&0,503 & 0,527 & 0,893 & 0,00 \\
    &&20&0,360 & 0,385 & 0,423 & 0,45 \\
    &&50&0,205 & 0,228 & 0,102 & 1,70 \\
    \hline
    
    \multirow{9}{*}{5.000}&\multirow{3}{*}{50}&50&0,199 & 0,224 & 0,074 & 1,94 \\
    &&100&0,116 & 0,135 & 0,039 & 2,19 \\
    &&250&0,054 & 0,063 & 0,022 & 2,09 \\
    \cline{2-7}
    &\multirow{3}{*}{100}&50&0,247 & 0,272 & 0,120 & 1,91 \\
    &&100&0,158 & 0,179 & 0,063 & 2,49 \\
    &&250&0,080 & 0,092 & 0,037 & 2,40 \\
    \cline{2-7}
    &\multirow{3}{*}{250}&50&0,290 & 0,313 & 0,205 & 1,60 \\
    &&100&0,203 & 0,223 & 0,105 & 2,45 \\
    &&250&0,116 & 0,129 & 0,065 & 2,58 \\
    \hline
    
    \multirow{9}{*}{10.000}&\multirow{3}{*}{100}&100&0,157 & 0,178 & 0,058 & 2,49 \\
    &&200&0,092 & 0,108 & 0,032 & 2,72 \\
    &&500&0,043 & 0,051 & 0,019 & 2,69 \\
    \cline{2-7}
    &\multirow{3}{*}{200}&100&0,194 & 0,214 & 0,090 & 2,50 \\
    &&200&0,123 & 0,140 & 0,052 & 2,96 \\
    &&500&0,063 & 0,072 & 0,033 & 2,72 \\
    \cline{2-7}
    &\multirow{3}{*}{500}&100&0,225 & 0,244 & 0,147 & 2,27 \\
    &&200&0,156 & 0,171 & 0,082 & 2,97 \\
    &&500&0,090 & 0,099 & 0,053 & 2,92 \\
    \hline
    \multicolumn{3}{|r|}{Average}&0,188 & 0,208 & 0,150 & 1,86 \\
    \hline
    \multicolumn{3}{|r|}{Number of Best Cases}&3 & 0 & 24 &\multicolumn{1}{r}{}\\
    \cline{1-6}
    \end{tabular}
\end{table}

\begin{table}[htbp]
    \caption{Average Results for Zipf Distribution with $a=2$}
    \scriptsize
    \centering
    \begin{tabular}{|c|c|c|c|c|c|c|}
    \hline
    \multicolumn{3}{|c|}{Test Parameters} & \multicolumn{3}{c}{Average Switch Ratio} & \multicolumn{1}{|c|}{Average Excess Run} \\
    \hline
    Input Size & Color Size & Buffer Size & BW & RC & Picky & Picky \\
    \hline
    \multirow{9}{*}{1.000}&\multirow{3}{*}{10}&10&0,287 & 0,308 & 0,243 & 0,20 \\
    &&20&0,162 & 0,182 & 0,074 & 0,56 \\
    &&50&0,073 & 0,083 & 0,029 & 0,54 \\
    \cline{2-7}
    &\multirow{3}{*}{20}&10&0,338 & 0,362 & 0,496 & 0,09 \\
    &&20&0,210 & 0,233 & 0,090 & 0,76 \\
    &&50&0,103 & 0,116 & 0,041 & 0,73 \\
    \cline{2-7}
    &\multirow{3}{*}{50}&10&0,372 & 0,397 & 0,771 & 0,04 \\
    &&20&0,248 & 0,267 & 0,127 & 0,72 \\
    &&50&0,135 & 0,148 & 0,063 & 0,84 \\
    \hline
    
    \multirow{9}{*}{5.000}&\multirow{3}{*}{50}&50&0,133 & 0,150 & 0,036 & 1,12 \\
    &&100&0,077 & 0,089 & 0,024 & 1,11 \\
    &&250&0,037 & 0,042 & 0,018 & 0,92 \\
    \cline{2-7}
    &\multirow{3}{*}{100}&50&0,146 & 0,161 & 0,047 & 1,16 \\
    &&100&0,090 & 0,101 & 0,032 & 1,21 \\
    &&250&0,047 & 0,052 & 0,026 & 0,95 \\
    \cline{2-7}
    &\multirow{3}{*}{250}&50&0,154 & 0,169 & 0,055 & 1,16 \\
    &&100&0,099 & 0,109 & 0,038 & 1,20 \\
    &&250&0,053 & 0,058 & 0,031 & 0,98 \\
    \hline
    
    \multirow{9}{*}{10.000}&\multirow{3}{*}{100}&100&0,089 & 0,100 & 0,024 & 1,30 \\
    &&200&0,052 & 0,060 & 0,018 & 1,24 \\
    &&500&0,026 & 0,029 & 0,015 & 0,95 \\
    \cline{2-7}
    &\multirow{3}{*}{200}&100&0,096 & 0,108 & 0,031 & 1,30 \\
    &&200&0,059 & 0,066 & 0,023 & 1,28 \\
    &&500&0,031 & 0,034 & 0,020 & 1,01 \\
    \cline{2-7}
    &\multirow{3}{*}{500}&100&0,101 & 0,112 & 0,034 & 1,35 \\
    &&200&0,064 & 0,070 & 0,026 & 1,28 \\
    &&500&0,034 & 0,038 & 0,022 & 0,99 \\
    \hline
    \multicolumn{3}{|r|}{Average}&0,123 & 0,135 & 0,091 & 0,93 \\
    \hline
    \multicolumn{3}{|r|}{Number of Best Cases}&2 & 0 & 25 & \multicolumn{1}{r}{}\\
    \cline{1-6}
    \end{tabular}
\end{table}

\end{subappendices}

\end{document}